\documentclass{acmsmall}
\usepackage{caption}

\let\0\undefined
\let\argmin\undefined
\let\argmax\undefined

\newcommand{\qedhere}{\qed}

\usepackage{footmisc}           
\usepackage[utf8]{inputenc}
\usepackage{amssymb}
\usepackage{amsmath}
\newtheorem{question}{Question}
\newtheorem{fact}{Fact}
\newtheorem{axiom}{Axiom}
\newtheorem{construction}{Construction}

\usepackage{etoolbox,xparse,xspace}
\usepackage{array,multirow,booktabs} 
\usepackage{bigstrut} 
\usepackage{dsfont}
\usepackage{tikz,pgfplots}
\usetikzlibrary{arrows,shapes,positioning}
\newcommand{\Comments}{1}
\definecolor{gray}{gray}{0.5}
\definecolor{darkgreen}{rgb}{0,0.5,0}
\newcommand{\mynote}[2]{\ifnum\Comments=1\textcolor{#1}{#2}\fi}

\newcommand{\tr}{\top}
\newcommand{\id}{\mathop{\mathrm{id}}}

\newcommand{\A}{\mathcal{A}}
\newcommand{\B}{\mathcal{B}}
\newcommand{\D}{\mathcal{D}}
\newcommand{\E}{\mathbb{E}}
\newcommand{\F}{\mathcal{F}}

\renewcommand{\O}{\mathcal{O}}
\renewcommand{\P}{\mathcal{P}}

\newcommand{\R}{\mathcal{R}}
\let\oldS\S 
\newcommand{\sect}{\mbox{\oldS\hspace{-.1mm}}}
\renewcommand{\S}{\mathcal{S}}
\newcommand{\T}{\mathcal{T}}
\newcommand{\V}{\mathcal{V}}

\newcommand{\X}{\mathcal{X}}

\renewcommand{\vec}[1]{{\mathbf{#1}}}

\newcommand{\0}{\vec{0}}

\renewcommand{\o}{\mathit{o}}

\newcommand{\convhull}{\mathsf{Conv}}
\newcommand{\conv}{\convhull}
\newcommand{\ext}{\mathrm{ext}}

\newcommand{\linear}{\mathsf{Lin}}
\newcommand{\affine}{\mathsf{Aff}}
\newcommand{\subgrad}[1]{d #1}

\newcommand{\selsubgrad}[2]{\in \partial {#1}}
\newcommand{\toto}{\rightrightarrows}

\newcommand{\lsc}{l.s.c.}
\newcommand{\dom}{\mathrm{dom}}
\newcommand{\defeq}{\doteq}
\newcommand{\ones}{\mathbbm{1}}
\newcommand{\abs}[1]{\left\lvert #1 \right\rvert}

\def\reals{\mathbb{R}}
\def\integers{\mathbb{Z}}
\def\extreals{\mathbb{\overline{R}}}

\newcommand{\argmin}{\mathop{\mathrm{argmin}}}
\newcommand{\argmax}{\mathop{\mathrm{argmax}}}

\newcommand{\argsup}{\mathop{\mathrm{argsup}}}

\newcommand{\inprod}[1]{\left\langle #1 \right\rangle}
\newcommand{\frobeniusprod}{\!:\!}

\newcommand{\shortcite}{\cite}

\newcommand{\scorename}{affine score\xspace}
\newcommand{\scorenames}{affine scores\xspace}
\newcommand{\scorevar}{\ensuremath{\mathsf{A}}}
\newcommand{\ASS}{\scorevar}
\newcommand{\SRS}{\ensuremath{\mathsf{S}}}

\newcommand{\MDS}{\ensuremath{\mathsf{M}}}
\newcommand{\AS}[2][]{\scorevar#1\parenargs{#2}}
\newcommand{\SR}[2][]{\SRS#1(#2)}

 \renewcommand{\ones}{\mathds{1}}
\renewcommand{\Comments}{0}
\renewcommand{\AS}[2][]{\ASS#1(#2)}
\newcommand{\ASdef}[1][]{\ASS#1 : \T \times \T \to \extreals}
\newcommand{\ASpropdef}[1][]{\ASS#1 : \R \times \T \to \extreals}

\title{General Truthfulness Characterizations Via Convex Analysis}
\author{
RAFAEL M. FRONGILLO\affil{University of Colorado, Boulder}
IAN A. KASH \affil{University of Illinois at Chicago}}

\begin{abstract}
We present a model of truthful elicitation which generalizes and extends mechanisms, scoring rules, and a number of related settings that do not qualify as one or the other.  Our main result is a characterization theorem, yielding characterizations for all of these settings, including a new characterization of scoring rules for non-convex sets of distributions.  We generalize this model to eliciting some property of the agent's private information, and provide the first general characterization for this setting.  We combine this characterization with duality to give a simple construction to convert between scoring rules and randomized mechanisms.  We also show how this characterization gives a new proof of a mechanism design result due to Saks and Yu.
\end{abstract}

\begin{document}

\maketitle

\section{Introduction}

Information elicitation, the gathering of information from an agent by a principal, is a key problem in economics, statistics, machine learning, and finance.  In these settings, one is interested in obtaining the preferences of an agent, a probability distribution from an expert, the desired prediction from an algorithm, and summary statistics of the risk of a portfolio, respectively.  They key challenge in all of these settings is to reward the agent in such a way that the agent will truthfully reveal his knowledge rather than being encouraged to reveal some incorrect version.  A central question in all these literatures has been to characterize all the ways this can be achieved; that is, to characterize all the truthful mechanisms, proper scoring rules, proper loss functions, or elicitable properties, respectively.  Many variants on such characterization theorems have been proved. (See Section~\ref{sec:prior} for a partial list.)  Moreover, the proofs of these results all use various tools from convex analysis; in particular, the characterizations tend to be in terms of properties of convex functions and their subgradients.

Despite this commonality of question and technique, the literature on mechanism design has proceeded essentially independent of these other literatures and vice versa.  Further, many characterizations are presented for a particular setting, and it is not immediately clear how they would apply to others.  As a result, there are many theorems in the literature whose proofs are slight variations on existing results to adapt them to a new setting.  While ex post the variations are slight, ex ante the needed changes were frequently not obvious and required significant effort to pin down.  An excellent example of this is the scoring rules characterization by Gneiting and Raftery~\citeyear{gneiting2007strictly}, which has had at least three variants used in various settings~\cite{boutilier2012eliciting,chen2011information,cid-sueiro2012proper}.

In this paper we address these two problems by formulating a model of information elicitation which is general enough to encompass all these variants, and providing a characterization for it.  This obviates the need for additional new theorems in any future application that fits into our framework.  Further, it provides a clearer understanding of the connections between mechanism design and scoring rules which allows us to translate results from one domain to the other.

Our model consists of a single agent endowed with some type $t$ known only to the agent, who is asked to reveal his type to the principal.  After doing so, the principal gives the agent a score $\AS{t',t}$ that depends on both the agent's reported type $t'$ and his true type $t$.
We allow \ASS\ to be quite general, with the main requirement being that $\AS{t',\cdot}$ is an affine
function (linear transformation plus a constant) of the true type $t$, and seek to understand when it is optimal for the agent to truthfully report his type.  Given this truthfulness condition, it is immediately clear why convexity plays a central role---when an agent's type is $t$, we desire the score for telling the truth to satisfy $\AS{t,t} = \sup_{t'} \AS{t',t}$, which means this ``consumer surplus'' function $G(t) := \AS{t,t}$ must be convex as the pointwise supremum of affine functions.

One special case of our model is mechanism design with a single agent,
where the designer wishes to select an outcome based on the agent's type.
In this setting, $\AS{t',\cdot}$ can be thought of as the allocation and payment given a report of $t'$, which combine to determine the utility of the agent as a function of his type.
In this context, $\AS{t,t}$ is the consumer surplus function (or indirect utility function), and Myerson's well-known characterization~\shortcite{myerson1981optimal} states that, in single-parameter settings, a mechanism is truthful if and only if the consumer surplus function is convex and its derivative (or subgradient at points of non-differentiability) is the allocation rule.  More generally, this remains true in higher dimensions (see~\cite{rochet1985taxation}).  Note that here the restriction that $\AS{t',\cdot}$ be affine is without loss of generality, because we view types as functions and function application is a linear operation. (See Section~\ref{chap:general-modelmechanism-design} for more details.)

Another special case is a scoring rule, also called a \emph{proper loss} in the machine learning literature, where an agent is asked to predict the distribution of a random variable and given a score based on the observed realization of that variable.  In this setting, types are distributions over outcomes, and $\AS{t',t}$ is the agent's subjective expected score for a report that the distribution is $t'$ when he believes the distribution is $t$.  As an expectation, this score is linear in the agent's type.
Gneiting and Raftery~\shortcite{gneiting2007strictly} unified and generalized existing results in the scoring rules literature by characterizing proper scoring rules in terms of convex functions
and their subgradients.

Further, the generality of our model allows it to include settings that do not quite fit into the standard formulations of mechanisms or scoring rules.
These include counterfactual scoring rules for decision-making~\cite{othman2010decision,chen2011information,chen2011decision}, proper losses for machine learning with partial labels~\cite{cid-sueiro2012proper}, mechanism design with partial allocations~\cite{cai2013designing}, responsive lotteries~\cite{feige2010responsive}, and mechanisms for crowdsourcing categorical information~\cite{shah2015double}.

In many settings, it is difficult, or even impossible, to have agents report an entire type $t\in\T$.  For example, when allocating a divisible good (e.g. water), a mechanism that needs to know how much an agent would value each possible allocation requires him to submit an infinite-dimensional type.  Even type spaces which are exponential in size, such as those that arise in combinatorial auctions, can be problematic from an algorithmic perspective.  Moreover, in many situations, the principal is \emph{uninterested} in all but some small aspect of an agent's private type.  For example, the information is often to be used to eventually make a specific decision, and hence only the information directly pertaining to the decision is actually needed---why ask for the agent's entire probability distribution of rainfall tomorrow if a principal wanting to choose between $\{$umbrella, no umbrella$\}$ would be content with its expected value, or even just whether she should carry an umbrella or not?

It is therefore natural to consider an indirect elicitation model where agents provide some sort of summary information about their type.  Such a model has been studied in the scoring rules literature, where one wishes to elicit some statistic, or \emph{property}, of a distribution, such as the mean or quantile~\cite{savage1971elicitation,osband1985information,lambert2008eliciting,gneiting2011making}.  We follow this line of research, and extend the \scorename framework to accept reports from a different (intuitively, much smaller) space than $\T$.

\subsection{Our Contribution}

Our main theorem (Theorem~\ref{thm:main-char}) is a general characterization theorem that generalizes and extends known characterization theorems for proper scoring rules (substantially) and truthful mechanisms (slightly).  We also survey applications to related settings and show our theorem can be used to provide characterizations for them as well, including new results about mechanism design with partial allocation and responsive lotteries.  Thus, our theorem eliminates the need to independently derive characterizations for such settings.  We further generalize our theorem to
the case where the desired report is a function of the private information (a {\em property}) rather than the full information (Theorem~\ref{thm:prop-char}) and provide a variant of this characterization which makes use of convex duality (Theorem~\ref{thm:dual-report-char}).  Finally, we conclude by examining cases where the set of possible reports is finite, which Lambert and Shoham~\citeyear{lambert2009eliciting} showed correspond to power diagrams, a generalization of Voronoi diagrams (see Section~\ref{sec:props-finite}).  We extend their result to settings where the private information need not be a probability distribution, and give a tight characterization for a particular restricted ``simple'' case.  We also give an explicit construction for generating power diagrams from other measures of distances via a connection to {\em Bregman Voronoi diagrams}~\cite{boissonnat2007bregman}.

\subsubsection{Scoring Rules}

Our contribution to the scoring rules literature is a characterization of proper scoring rules for non-convex sets of distributions, the first of its kind.  As motivation, note that it is very natural to ask for scoring rules which are proper with respect to a particular family of distributions, e.g. if the principal was convinced that the agent's belief came from such a family, but many common families are non-convex, including most exponential families (Gaussian, Poisson, exponential, Laplace, Pareto, etc.)\,\cite{nielsen2009statistical}.  Moreover, general non-convex sets of distributions have proven useful as a way of separating informed and uninformed experts~\cite{babaioff2011only,fang2010proper}.  Despite how natural and useful it is to restrict to the non-convex case, no characterization was known.

We give such a characterization, which surprisingly shows that the only scoring rules which are proper for a non-convex $\P$ are those which can be extended to a proper scoring rule on $\conv\P$; in other words, one does not gain flexibility by ruling out distributions in the ``interior'' of $\P$.  Interestingly, the main technical tool we need to extend the proof to this case comes from the mechanism design literature, where characterizations for non-convex type spaces have been previously established.  Additionally, we show that properness of a scoring rule is a local property, in the sense that it suffices to verify it in a neighborhood around each distribution. See Section~\ref{sec:local} and Corollary~\ref{thm:locally-proper}. 

\subsubsection{Property Elicitation}

The two subsequent generalizations of our main theorem provide the first general characterizations for arbitrary elicitable properties, which capture only part of the agent's private information (e.g. the mean of a probability distribution rather than the distribution itself).  The first is essentially a direct generalization of Theorem~\ref{thm:main-char}, which keeps the same general structure but adds the constraint that the convex function must be flat on sets of types which share an optimal report. In addition to serving as our main tool to derive the remainder of our results, this theorem provides several ways to show that a property is not elicitable (by showing that no such convex function can exist).  This allows us to show that the smallest confidence interval containing a given amount of mass is not elicitable, thus settling an open question.

The second result is a transformation of this theorem using {\em duality}, which shows that there is a strong sense in which properties {\em are} subgradients of convex functions.  Relatedly, we introduce the notion of {\em direct elicitability}, where the the property to be elicited is the subgradient rather than a link function applied to it.  We also use this result to introduce notions of dual properties and scores, which gives a new construction to convert between scoring rules and randomized mechanisms (see Construction~\ref{construction}).
Subsequent work has made use of these theorems to derive new results on linear and vector-valued properties~\cite{frongillo2015vector}, elicitation complexity~\cite{frongillo2015elicitation}, and peer prediction~\cite{frongillo2016geometric}.

\subsubsection{Mechanism Design}

For mechanism design, our contributions are more modest.  Our characterization is a minor extension of Archer and Kleinberg's characterization, removing a technical assumption~\cite[Theorem 6.1]{archer2008truthful}.  We show how many previous results about implementability and revenue equivalence can be translated into our framework, but do not introduce significant new results.  Instead, the main interest of our approach is that by translating economic questions into convex analysis questions, we can simplify some of the proofs and expose the underlying intuition.  Additionally, we show how known results about scoring rules yield a new proof of an implementability theorem due to Saks and Yu~\citeyear{saks2005weak}.

\subsubsection{Novel Elicitation Settings}

Perhaps the most useful direct application of our main theorem is to elicitation settings that do not quite match the standard frameworks of scoring rules or mechanism design, as it immediately provides a characterization for such settings.  In Section~\ref{sec:other-applications} we demonstrate the versatility of our characterization by surveying five recent such examples (on decision rules, proper losses for partial labels, mechanism design with partial allocations, responsive lotteries, and crowdsourcing), and showing how our results could have been applied.

\subsubsection{Summary of Novel Results}

Since our results cover a range of applications and include many reframings or small extensions of existing results, we summarize the main novel results here.

\begin{enumerate}
\item A general characterization theorem for many non-standard applications;
\item A characterization of scoring rules for non-convex sets of distributions;
\item Scoring rules are proper iff they are locally proper;
\item A new geometric proof of the Saks--Yu~\citeyear{saks2005weak} result on implementable mechanisms;
\item A characterization of elicitability for arbitrary properties;
\item The smallest confidence interval is not elicitable;
\item A new construction to convert between scoring rules and randomized mechanisms.
\end{enumerate}

\subsection{Related Work}
\label{sec:prior}

The similarities between mechanisms and scoring rules were noted by (among others) Fiat et al.~\shortcite{fiat2013approaching}, who gave a construction to convert mechanisms into scoring rules and vice versa, and Feige and Tennenholtz~\shortcite{feige2010responsive}, who gave techniques to convert both to ``responsive lotteries.''  Further, techniques from convex analysis have a long history in the analysis of both models (see~\cite{gneiting2007strictly,vohra2011mechanism}).  However, we believe that our results use the ``right'' representation and techniques, which leads to more elegant characterizations and arguments.  For example, the construction used by Fiat et al. has the somewhat awkward property that the scoring rule corresponding to a mechanism has one more outcome than the mechanism did, a complication absent from our results.
Similarly, the constructions used by Feige and Tennenholtz only handle special cases and they claim ``there is no immediate equivalence between lottery rules and scoring rules,'' while we can give such an equivalence.
So while prior work has understood that there is a connection, the nature of that connection has been far from clear.

A large literature in mechanism design has explored characterizations of when allocation rules can be truthfully implemented; see e.g.~\cite{mcafee1988multidimensional,jehiel1996how,jehiel1999multidimensional,jehiel2001efficient,saks2005weak,bikhchandani2006weak,muller2007weak,archer2008truthful,ashlagi2010monotonicity,carroll2012when}.
Similarly, work on revenue equivalence can be cast in our framework as well
\cite{myerson1981optimal,krishna2001convex,heydenreich2009characterization,carbajal2012mechanism}.
For scoring rules, our work connects to a literature that has used non-convex sets of probability distributions to separate (usefully) informed exports from uninformed experts
\cite{fang2010proper,babaioff2011only}.

The study of indirect elicitation in scoring rules can be traced to Savage in 1971, who considered the problem of eliciting expected values of random variables~\cite{savage1971elicitation}.  Osband and Reichelstein~\citeyear{osband1985information} go on to provide a rigorous version, generalizing to expected values of functions of the underlying variable.  Since then several similar results have appeared with varying assumptions (e.g.~\cite{banerjee2005optimality,gneiting2011making,abernethy2012characterization}), all of which have recently been unified in~\cite{frongillo2015vector}.  More generally, many authors have considered other common statistics as well, such as quantiles, ratios of expectations, and expectiles~\cite{osband1985information,gneiting2007strictly,gneiting2011making,grant2013consistent}.

While these and many other examples of specific statistics have appeared in the literature, it was perhaps Osband~\citeyear{osband1985providing} and Lambert, Pennock, and Shoham~\citeyear{lambert2008eliciting} who first considered the following general problem: given an outcome space $\O$ and an \emph{arbitrary} map $\Gamma : \Delta(\O) \to \reals$, under what circumstances can we construct a proper scoring rule to elicit $\Gamma(t)$?  Moreover, what is the full classification of functions $\Gamma$ which can be elicited in this way?

Several authors have made significant contributions toward answering this general question for the case where $\Gamma$ is real-valued~\cite{lambert2008eliciting,lambert2018elicitation,gneiting2011making,gneiting2014probabilistic,steinwart2014elicitation} and vector-valued~\cite{osband1985providing,lambert2008eliciting}.
Lambert and Shoham~\citeyear{lambert2009eliciting} also characterized elicitable properties $\Gamma$ which take on finitely many values, showing a connection to power diagrams from computational geometry.

We focus primarily on properties in the scoring rule context because, while non-direct-revelation mechanisms are often studied, the focus is not typically on the report itself but on the outcome of the mechanism.  
However, properties are somewhat more natural in mechanism design with a finite set of allocations.  In particular, mechanisms that elicit a ranking over outcomes rather than a utility for each outcome (common in, e.g., matching contexts) are a form of property elicitation, and our results are related to characterizations due to Carroll~\citeyear{carroll2012when}.  Our results about finite properties also provide a new proof of a theorem due to Saks and Yu that characterizes when allocation rules that select from a finite set of allocations have payments that make them truthful~\cite{saks2005weak}.

Subsequent to our work, our properties characterizations have been applied in several contexts.  They have been used to provide a new unified characterization of linear properties and several results towards to a characterization of vector-valued properties~\cite{frongillo2015vector}, as the basis to move from the question of whether a property is elicitable to how difficult it is to elicit~\cite{frongillo2015elicitation}, and to characterize minimal peer prediction mechanisms~\cite{frongillo2016geometric}.

\subsection{Notation}
We define $\extreals = \reals\cup\{-\infty,\infty\}$ to be the extended real numbers.
Given a set of measures $M$ on a space $X$ with $\sigma$-algebra $\B$, a function $f:X\to \extreals$ is $M$-quasi-integrable if $\int_X f(x) d\mu(x) \in \extreals$ for all $\mu\in M$.  Let $\Delta(X)$ be the set of all probability measures on $X$.  We denote by $\affine(X\to Y)$ and $\linear(X\to Y)$ the set of functions from $X \subseteq \V_1$ to $Y \subseteq \V_2$ which are restrictions of affine and linear functions (respectively) from vector space $\V_1$ to vector space $\V_2$.  We write $\convhull(X)$ to denote the convex hull of a set of vectors $X$, the set of all (finite) convex combinations of elements of $X$. Some useful facts from convex analysis are collected in Appendix~\ref{sec:cvx-primer}.

\section{Affine Scores}
\label{sec:main}

We consider a very general model with an agent who has a given type $t\in\T$ and reports some possibly distinct type $t'\in\T$, at which point the agent is rewarded according to some score $\AS{t',t}$ which is affine in the true type $t$.  This reward we call an \scorename.  We wish to characterize all \emph{truthful} \scorenames, which incentivize the agent to report her true type $t$.

\begin{definition}
  \label{def:affine-score}
Let $\T\subseteq\V$ for some vector space $\V$ over $\reals$.  A function $\ASdef$ is an \emph{\scorename} with \emph{score set} $\A \defeq \{ \AS{t,\cdot} ~|~ t \in \T \} \subseteq \affine(\T\to\extreals)$.\footnote{To define linear functions to $\extreals$, we adopt the convention $0\cdot\infty = 0\cdot(-\infty) = 0$.  Thus, any $\ell \in \linear(\V\to\extreals)$ can be written as $\ell_1 + \infty \cdot \ell_2$ for some $\ell_1,\ell_2 \in \linear(\V\to\reals)$.}  We say $\ASS$ is \emph{truthful} if for all $t,t' \in \T$,
  \begin{equation}
    \AS{t',t} \leq \AS{t,t}.
  \end{equation}
  If this inequality is strict for all $t \neq t'$, then \ASS\ is \emph{strictly truthful}.
\end{definition}

Our characterization uses convex analysis, a central concept of which is the subgradient of a function, which is a generalization of the gradient yielding a linear approximation that is always below the function.

\begin{definition}
  Given some function $G:\T\to\reals$, a function $d \in \linear(\V\to\extreals)$ is a \emph{subgradient} to $G$ at $t$ if for all $t'\in \T$,
  \begin{equation}
    \label{eq:subgradient}
    G(t') \geq G(t) + d(t'-t).
  \end{equation}
  We denote by $\partial G:\T\toto\linear(\V\to\extreals)$ the multivalued map such that $\partial G_t$ is the set of subgradients to $G$ at $t$.
  \footnote{It is important to note a deviation from standard terminology here: typically, a ``vertical'' subgradient, which can take value $\pm\infty$, is called a \emph{subtangent} (cf.~\cite{gneiting2007strictly}); here we use \emph{subgradient} for both cases.}
  We will occasionally overload the $\partial G$ notation to mean $\partial G = \cup_{t\in\T} \partial G_t$.
We say a parameterized family of linear functions $\{ d_t \in \linear(\V\to\extreals) \}_{t  \in \T'}$ for $\T'\subseteq\T$ is a \emph{selection of subgradients} if $d_t \in \partial G_t$ for all $t\in\T'$; we denote this succinctly by $\{ d_t \}_{t  \in \T'} \selsubgrad{G}{\T'}$.
\end{definition}

For mechanism design, it is typical to assume that utilities are always real-valued.  However, the log scoring rule (one of the most popular scoring rules) has the property that if an agent reports that an event has probability 0, and then that event does occur, the agent receives a score of $-\infty$.  Essentially solely to accommodate this, we allow \scorenames and subgradients to take on values from the extended reals.  In the next paragraph we provide the relevant definitions, but for most purposes it suffices to ignore these and simply assume that all \scorenames are real-valued.

It is standard (cf.~\cite{gneiting2007strictly}) to restrict consideration to the ``regular'' case, where intuitively only things like the log score are permitted to be infinite.  In particular, an \scorename $\ASS$ is \emph{regular} if $\AS{t,t} \in \reals$ for all $t \in \T$, and $\AS{t',t} \in \reals \cup \{-\infty\}$ for $t' \neq t$.  Similarly, a parameterized family of linear functions (e.g. a family of subgradients) $\{ d_t \in \linear(\V\to\extreals) \}_{t  \in \T}$ is \emph{$\T$-regular} if $d_t(t) \in \reals$ for all $t \in \T$, and $d_{t'}(t) \in \reals \cup \{-\infty\}$ for $t' \neq t$.  Likewise, $\T$-regular affine functions have $\T$-regular linear parts with finite constants (i.e. we exclude the constant functions $\pm \infty$).  For the remainder of the paper we assume all \scorenames and parameterized families of linear or affine functions are $\T$-regular, where $\T$ will be clear from context.

We now state, and prove, our characterization theorem.  The proof takes Gneiting and Raftery's~\shortcite{gneiting2007strictly} proof for the case of scoring rules on convex domains and extends it to the non-convex case using a variant of a technique Archer and Kleinberg~\shortcite{archer2008truthful} introduced for mechanisms with non-convex type spaces.  This technique is essentially that used in prior work on extensions of convex functions~\cite{peters1987convex,yan2012extension}.
\begin{theorem}
  \label{thm:main-char}
  Let an \scorename $\ASdef$ with score set $\A$ be given.  \ASS\ is truthful if and only if there exists some convex $G : \convhull(\T)\to\extreals$ with $G(\T) \subseteq\reals$, and some selection of subgradients $\{d_t\}_{t\in\T} \selsubgrad{G}{\T}$, such that
  \begin{equation}
    \label{eq:main-char}
    \AS{t',t} = G(t') + d_{t'}(t-t').
  \end{equation}
\end{theorem}
\begin{proof}
  It is immediate from the subgradient inequality~\eqref{eq:subgradient} that the proposed form is in fact truthful, as
\[\AS{t',t} = G(t') + d_{t'}(t-t') \leq G(t) = G(t) + d_t(t-t) = \AS{t,t}.\]
For the converse, we are given some truthful $\ASdef$.

By definition, each $\AS{t,\cdot}$ is the restriction of an affine function defined on all of $\V$, and we now argue that $\AS{t,\cdot}$ is therefore well-defined on $\convhull(\T)$.
While in general there may be many different affine functions on $\V$ whose restriction to $\T$ is $\AS{t,\cdot}$,
given any $\hat t \in \convhull(\T)$ we may represent $\hat t$ as a finite convex combination $\hat t = \sum_{i=1}^m \alpha_i t_i$ where $t_i \in \T$.
Thus by their affineness all of these yield the same values on $\convhull(\T)$ as one easily sees that $\AS{t,\hat t} = \sum_{i=1}^m \alpha_i \AS{t,t_i}$.
Thus, $\AS{t,\cdot}$ is both determined and well-defined on $\convhull(\T)$.

  Now we let $G(\hat t) \defeq \sup_{t\in\T} \AS{t,\hat t}$, which is convex as the pointwise supremum of convex (in our case affine) functions.  Since \ASS\ is truthful, we in particular have $G(t) = \AS{t,t} \in \reals$ for all $t\in\T$ by our regularity assumption.  Let  $\AS[_\ell]{t,\cdot}$ denote the linear part of $\AS{t,\cdot}$.  Then, also by truthfulness, we have for all $t'\in \T$ and $\hat t \in \convhull(\T)$,
  \begin{align*}
    G(\hat t)
    \defeq \sup_{t\in\T} \AS{t,\hat t}
    \geq \AS{t',\hat t}
    = \AS{t',t'} + \AS[_\ell]{t',\hat t-t'}
    = G(t') +  \AS[_\ell]{t',\hat t-t'}.
  \end{align*}
Hence, $\AS[_\ell]{t',\cdot}$ satisfies~\eqref{eq:subgradient} for $G$ at $t'$, so $\ASS$ is of the form~\eqref{eq:main-char}.
\end{proof}

In the remainder of this section, we show how scoring rules, mechanisms, and other related models fit comfortably within our framework.

\subsection{Scoring Rules for Non-Convex $\P$}
\label{sec:scoring}

In this section, we show that the Gneiting and Raftery characterization is a simple special case of Theorem~\ref{thm:main-char}, and moreover that we \emph{generalize} their result to the case where the set of distributions $\P$ may be non-convex.  We also give a result about local properness derived using tools from mechanism design in Appendix~\ref{sec:local}.  To begin, we formally introduce scoring rules and show that they fit into our framework.  The goal of a scoring rule is to incentivize an expert who knows a probability distribution to reveal it to a principal who can only observe a single sample from that distribution.

\begin{definition}
  Given outcome space $\O$ and set of probability measures $\P \subseteq \Delta(\O)$, a \emph{scoring rule} is a function $\SRS:\P\times\O\to \extreals$ such that $\SR{p,\cdot}$ is $\P$-quasi-integrable for all $p \in \P$ (see below).  We say $\SRS$ is \emph{proper} if for all $p,q \in \P$,
  \begin{equation}
    \label{eq:gen-model-proper}
    \E_{\o\sim p} [ \SR{q,\o} ] \leq \E_{\o\sim p} [ \SR{p,\o} ].
  \end{equation}
  If the inequality in~\eqref{eq:gen-model-proper} is strict for all $q \neq p$, then $\SRS$ is \emph{strictly proper}.
\end{definition}

To incorporate this into our framework, take the type space $\T = \P$.  Thus, we need only construct the correct score set of affine functions available to the scoring rule as payoff functions.  Intuitively, these are the functions that describe what payment the expert receives given each outcome, but we have a technical requirement that the expert's expected utility be well defined.  Thus, following Gneiting and Raftery, we take $\F$ to be the set of $\P$-quasi-integrable
functions $f:\O\to\extreals$,
meaning $\int_\O f(\o)dp(\o) \in \extreals$ for all $p\in\P$,
and the score set $\A = \{p \mapsto  \E_{\o\sim p} \SR{q,\o} \mid q \in P\} \subseteq \{p\mapsto \int_\O f(\o)\,dp(\o) \mid f \in \F\}$.
Note that in this case $\A$ actually contains \emph{linear} functions of $p$, which are trivially affine.

We now apply Theorem~\ref{thm:main-char}, which yields the following generalization of Gneiting and Raftery~\shortcite{gneiting2007strictly}.
\begin{corollary}
  \label{thm:scoring-rule-char}
  For an arbitrary set $\P \subseteq \Delta(\O)$ of probability measures, a regular\footnote{This is the same concept as with \scorenames: scores cannot be $\infty$ and only incorrect reports can yield $-\infty$.} scoring rule $\SRS:\P\times\O\to\extreals$ is proper if and only if there exists a convex function $G:\convhull(\P)\to\reals$ with functions $G_p \in \F$ such that
  \begin{equation}
    \label{eqn:sr}
    \SR{p,\o} = G(p) + G_p(\o) - \int_\O G_p(\o')\,dp(\o'),
  \end{equation}
  where $G_p : q\mapsto \int_\O G_p(\o')\,dq(\o')$ is a subgradient of $G$ for all $p\in\P$.
\end{corollary}
\begin{proof}
  Truthfulness of the given form follows immediately from Theorem~\ref{thm:main-char} and our definition of $\A$.
  For the converse, let $\ASdef$ be a given truthful \scorename.
  From the theorem, $\AS{p,\cdot} = G(p) + d_{p}(\,\cdot\, - p) \in \A$, so we must have $d_p \in \F$; the subgradients are then of the form $G_p : q\mapsto \int_\O d_p(\o)\,dq(\o)$ as desired.
\end{proof}

Importantly, Corollary~\ref{thm:scoring-rule-char} immediately generalizes the characterization of Gneiting and Raftery~\citeyear{gneiting2007strictly} to the case where $\P$ is not convex, which is new to the scoring rules literature.  One direction of this extension is obvious (if $\SRS$ is truthful on the convex hull of a set then it is truthful on that set), but the other is not, and is an important negative result in that it rules out the possibility of new scoring rules arising by restricting the set of distributions (as long as the restriction does not change the convex hull of the set).

In the absence of a characterization, several authors have worked in the non-convex $\P$ case.  For example, Babaioff et al.~\shortcite{babaioff2011only} examine when proper scoring rules can have the additional property that uninformed experts do not wish to make a report (have a negative expected utility), while informed experts do wish to make one.  They show that this is possible in some settings where the space of reports is not convex.  Our characterization shows that, despite not needing to ensure properness on reports outside $\P$, essentially the only possible scoring rules are still those that are proper on all of $\Delta(\O)$.  We state the simplest version of such a characterization, for perfectly informed experts, here.

\begin{corollary}
\label{cor:valuable}
Let a non-convex set $\P \subseteq \Delta(\O)$ and $\bar{p} \in \Delta(\O) - \P$ be given.  A scoring rule $\SRS$ is proper and guarantees that experts with a belief in $\P$ receive a score of at least $\delta_A$ while experts with a belief of $\bar{p}$ receive a score of at most $\delta_R$ if and only if $\SRS$ is of the form~\eqref{eqn:sr} with $G(p) \geq \delta_A \forall p \in \P$ and $G(\bar{p}) \leq \delta_R$.
\end{corollary}

With a similar goal to Babaioff et al., Fang et al.~\shortcite{fang2010proper} find conditions on $\P$ for which every continuous ``value function'' $G:p\mapsto \SR{p,p}$ on $\P$ can be attained by some $\SRS$ with the motivation of eliciting the expert's information when it is known to come from some family of distributions (which in general will not be a convex set).    As such, they provide sufficient conditions on particular non-convex sets, as opposed to our result which provides necessary and sufficient conditions for all non-convex sets.  Beyond these specific applications, our characterization is useful for answering practical questions about scoring rules.  For example, suppose we assume that people have beliefs about probabilities in increments of 0.01.  Does that change the set of possible scoring rules? No. What happens if they have finer-grained beliefs but we restrict them to such reports? They will end up picking a ``nearby'' report (see the discussion of convexity in Section~\ref{sec:notelicitable}).

In Appendix~\ref{sec:local}, we show how local truthfulness conditions, where one verifies that an \scorename is truthful by checking that it is truthful in a small neighborhood around every point, from mechanism design generalize to our framework.  In particular Corollary~\ref{cor:wlt} shows that local properness (i.e. properness for distributions in a neighborhood) is equivalent to global properness for scoring rules on convex $\P$, an observation that is also new to the scoring rules literature.  See Appendix~\ref{sec:local} for the precise meaning of (weak) local properness (i.e. truthfulness).

\begin{corollary}
  \label{thm:locally-proper}
  For a convex set $\P \subseteq \Delta(\O)$ of probability measures, a scoring rule $\SRS:\P\times\O\to\extreals$ is proper if and only if it is (weakly) locally proper.
\end{corollary}

\subsection{Mechanism Design}
\label{chap:general-modelmechanism-design}

We now show how to view a mechanism as an \scorename.  First, we formally introduce mechanisms in the single agent case. (See below for remarks about multiple agents.)  Then we show how known characterizations of truthful mechanisms follow easily from our main theorem.  This allows us to relax a minor technical assumption from the most general such theorem.

\begin{definition}
\label{def:mechanism}
  Given outcome space $\O$ and a type space $\T \subseteq (\O\to\reals)$, consisting of functions mapping outcomes to reals, a (direct) \emph{mechanism} is a pair $(f,p)$ where $f:\T\to\O$ is an \emph{allocation rule} and $p:\T\to\reals$ is a \emph{payment}.  The utility of the agent with type $t$ and report $t'$ to the mechanism is $U(t',t) = t(f(t')) - p(t')$; we say the mechanism $(f,p)$ is \emph{truthful} if $U(t',t) \leq U(t,t)$ for all $t,t' \in \T$.
\end{definition}

Here we suppose that the mechanism can choose an allocation from some set $\O$ of outcomes, and there is a single agent whose type $t\in\T$ is itself the valuation function.  That is, the agent's net utility upon allocation $\o$ and payment $p$ is $t(\o) - p$.  Thus, following Archer and Kleinberg~\citeyear{archer2008truthful}, we view the type space $\T$ as lying in the vector space $\V=\reals^{\O}$.  The advantage of this representation is that while agent valuations in mechanism design can generally be complicated functions, viewed this way they are all linear: for any $v_1,v_2 \in \V$, we have $(v_1 + \alpha v_2)(\o) = v_1(\o) + \alpha v_2(\o)$.
Thus, we have an \scorename $\AS{t',t} \defeq U(t',t)$, with score set $\A = \{U(t',\cdot) \,|\, t' \in \T\} \subset \{t \mapsto t(\o) + c\,|\, \o\in\O, c\in\reals\}$, so that every combination of outcome and payment a mechanism could choose is potentially an element of $\A$.

As an illustration of our theorem, consider the following characterization, due to Myerson~\shortcite{myerson1981optimal}, for a single parameter setting (i.e. when the agent's type can be described by a single real number).  The result states that an allocation rule is implementable, meaning there is some payment rule making it truthful, if and only if it is \emph{monotone} in the agent's type.

\begin{corollary}[Myerson~\shortcite{myerson1981optimal}]
  \label{cor:myerson}
  Let $\T = \reals_+$, $\O \subseteq \reals$, so that the agent's valuation is $t \cdot \o$.  Then a mechanism $f,p$ is truthful if and only if
  \begin{enumerate}
  \item $f$ is monotone non-decreasing in $t$,
  \item $p(t) = t f(t) - \int_{0}^t f(t') dt' + p_0$.
  \end{enumerate}
\end{corollary}
\begin{proof}
  By elementary results in convex analysis $f$ is a subgradient of a convex function on $\reals$ if and only if it is monotone non-decreasing.  By  Theorem~\ref{thm:main-char}, the mechanism is truthful if and only if $f$ is a subgradient of the particular function $G(t) = U(t,t) = t(f(t)) - p(t)$, which is equivalent to (i) and the condition $G(t) = \int_{0}^t f(t')dt' + C$.
\end{proof}

More generally, applying our theorem gives the following characterization.  It is essentially equivalent to that of Archer and Kleinberg~\shortcite{archer2008truthful} (their Theorem 6.1), although our approach allows the relaxation of a technical assumption they term ``outcome compactness'' which their version requires when the set of types is non-convex.

\begin{corollary}
  \label{cor:archerkleinberg}
A mechanism $f,p$ is truthful if and only if there exists a convex function $G : \convhull(\T)\to\reals$ and some selection of subgradients $\{\subgrad{G}_t\}_{t\in\T}$, such that for all $t \in \T$, $f(t) = \subgrad{G}_t$ and $G(t) = t(f(t)) - p(t)$.
\end{corollary}

Note that we have written $f(t) = \subgrad{G}_t$, but the allocation is an outcome $f(t) \in \O$ while the subgradient is a linear function $\subgrad{G}_t \in (\O \to \reals) \to \reals$.
Recall the natural isomorphism between outcomes $o$ and functions $t \mapsto t(o)$, which means that we can always represent an outcome as a function sending types to their value for that outcome.
In that sense, we mean that $\subgrad G_t$ is the function $\subgrad G_t : t' \mapsto t'(o)$ where $o = f(t)$.
Note that the intuition that the subgradient is the allocation function still holds up to this technicality.

While we have thus far dealt with single-agent deterministic mechanisms implemented in dominant strategies, this characterization actually applies significantly more broadly.  In a sense, extending our characterizations to multiple agents is trivial: a mechanism is truthful if and only if it is truthful for agent $i$ when fixing the reports of the other agents.  Hence, we merely apply our characterization to each single-agent mechanism induced by reports of the other agents.  This is sufficient for our present study, but there are certainly reasons to take a more nuanced approach to the multi-agent setting---see Section~\ref{sec:discussion} for further discussion.  To extend to randomized mechanisms, one can take $f:\T\to\Delta(\O)$ and define $U(t',t) = \E_{\o\sim f(t')}[t(\o)] - p(t')$, which is still affine in $t$.  We can even extend to non-risk-neutral agents by taking the outcome space to be $\O' \defeq \Delta(\O)$. Finally, we can extend to Bayesian agents; in the above discussion of the multi-agent setting, take expectations instead of fixing specific types for the other agents.

Of course, mechanism design asks many questions beyond whether a particular mechanism is truthful, and some of these can be reframed as questions in convex analysis.  Implementability focuses on the question of when there exist payments that make a given allocation rule truthful.  
Figure~\ref{fig:oldnewproofs} (a) illustrates known characterizations and how they were proved.  As it shows, several of them rely on showing equivalence to a condition from convex analysis known as {\em cyclic monotonicity}.  Instead, in Appendix~\ref{sec:convex}, we reprove them in our more general framework by showing equivalence to the condition of being subgradients of a convex function (see Figure~\ref{fig:oldnewproofs} (b)).  This has three main benefits.
First, by exposing the essential convex analysis question, we are able to {\em greatly} simplify the proofs of some of these results.  For example, the original proof of Theorem~\ref{thm:wmon} relies on representing the allocation rule using a graph and arguing about the limit behavior of a process of creating paths in that graph.  In contrast, our proof simply requires defining a function and showing it is convex with the correct subgradients by elementary arguments.  
Second, this reframing reveals that these results actually yield, we believe, new results in convex analysis (in particular, Theorems~\ref{thm:cmon},~\ref{thm:wmon}, and~\ref{thm:wlsg} and Corollary~\ref{cor:lwmon}).
Third, this approach shows us how to translate known results from mechanism design into new results about scoring rules, as we saw in Section~\ref{sec:scoring}.
While elements of a subgradient-based approach can be found in a variety of work on characterizing implementability (see, e.g.,\cite{mcafee1988multidimensional,jehiel1996how,jehiel1999multidimensional,jehiel2001efficient,krishna2001convex,milgrom2002envelope,bikhchandani2006weak}),
this work has tended to use individual facts applied to particular settings, in contrast to our approach of translating mechanism design questions into convex analysis questions.  Nevertheless, as these are essentially reframings of known results that do not directly provide new insights for mechanism design, we defer this material to Appendix~\ref{sec:convex}.

\tikzset{
 char/.style={rectangle, draw=black, outer sep=4pt, align=center},
 who/.style={outer sep=6pt},
 iff/.style={implies-implies,double equal sign distance, thick, draw}
}

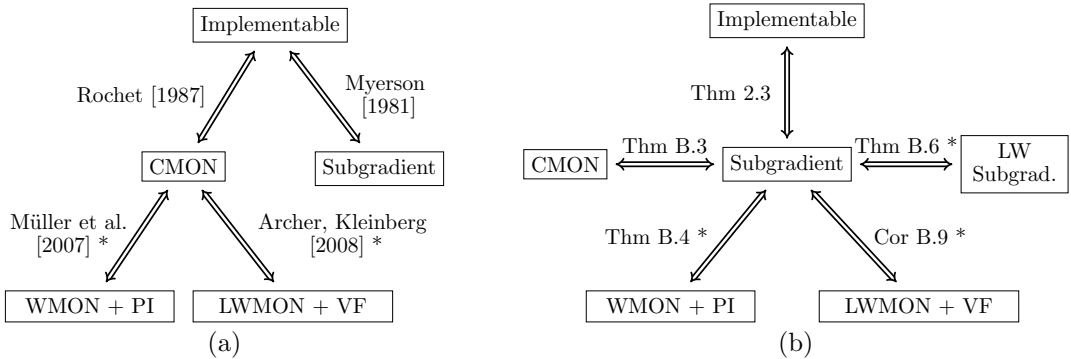
\begin{figure}[!h]
  \hspace*{-10pt}
  \begin{tabular}{ccc}
    \scalebox{0.82}{\begin{tikzpicture}[node distance=1.5cm and -0.8cm,
        every edge/.style={iff}]
        \node[char,rectangle] (imp) {Implementable};
        \node[char,rectangle] (cmon) [below left= of imp] {CMON} edge node[left,who] {Rochet~\shortcite{rochet1987necessary}} (imp);
        \node[char,rectangle] (sg) [below right= of imp] {Subgradient} edge node[right,who] {\shortstack{Myerson\\[-2pt]~\shortcite{myerson1981optimal}}} (imp);
        \node[char] (wmon) [below left= of cmon, text width=2.5cm] {WMON + PI} edge node[left,who] {\shortstack{\!\!\!M\"uller et al.\!\!\!\\~\shortcite{muller2007weak} *}}(cmon);
        \node[char] (lwmon) [below right= of cmon, text width=3cm] {LWMON + VF} edge node[right,who] {\shortstack{Archer, Kleinberg\\[-2pt]~\shortcite{archer2008truthful} *}} (cmon);
      \end{tikzpicture}}
 &\hspace*{6pt}&
   \scalebox{0.82}{\begin{tikzpicture}[node distance=1.5cm and -0.8cm, every edge/.style={iff}]
       \node[char,rectangle] (imp) {Implementable};
       \node[char,rectangle] (sg) [below = 1.5cm of imp] {Subgradient} edge node[left=-2pt,who] {Thm~\ref{thm:main-char}} (imp);
       \node[char] (wmon) [below left= of sg, text width=2.5cm] {WMON + PI} edge node[left,who] {Thm~\ref{thm:wmon} *} (sg);
       \node[char] (lwmon) [below right= of sg, text width=3cm] {LWMON + VF} edge node[right,who] {Cor~\ref{cor:lwmon} *} (sg);
       \node[char] (wlsg) [right= 1.5cm of sg, text width=1.5cm] {\centering LW Subgrad.} edge node[above=-4pt,who] {Thm~\ref{thm:wlsg} *} (sg);
       \node[char,rectangle] (cmon) [left= 1.6cm of sg] {CMON} edge node[above=-4pt,who] {Thm~\ref{thm:cmon}} (sg);
     \end{tikzpicture}}
\\
(a) && (b)    
  \end{tabular}
 \caption{Proof structure of existing mechanism design literature (a), and the new proof structure presented in this paper (b).  Asterisks (*) denote the requirement that $\T$ be convex.  We write CMON for cyclic monotonicity, WMON for weak monotonicity, and PI for path independence.  For the other abbreviations, L is local, W is weak, and VF is vortex freeness, a condition weaker than path independence introduced in Archer and Kleinberg~\protect\citeyear{archer2008truthful}.  Definitions of these conditions can be found in Appendix~\ref{sec:convex}.
}
 \label{fig:oldnewproofs}
\end{figure}

Revenue equivalence is the question of when all mechanisms with a given allocation rule charge the same prices (up to a constant).  Translating this into convex analysis terms, when is the convex function associated with a given selection of a subgradient unique up to a constant?  We ask the more general question, what are all the convex functions consistent with a given selection of a subgradient? The result is a theorem, extending a result due to Kos and Messner~\shortcite{kos2012extremal}, that characterizes the possible payments of every truthful mechanism, even those that do not satisfy revenue equivalence.  As their analysis essentially applies the natural convex analysis technique, we again defer this material to Appendix~\ref{sec:rev-eq}.

\section{Additional Models as Special Cases}
\label{sec:other-applications}

A number of other application domains do not quite fit into mechanism design or scoring rules, forcing researchers to adapt results to their particular setting.
For example, one may wish to elicit several distributions at once, or partially allocate items or rewards.
Fortunately, many such settings can be easily expressed in terms of the more general framework of affine scores.
We now briefly survey five such domains, and in each show how our main theorem could have directly provided the characterization ultimately used, rather than requiring effort to conceptualize and prove it.

\subsection{Decision Rules}

A line of work has considered a setting where a decision maker needs to select from a finite set $\mathcal{D}$ of decisions and so desires to elicit the distribution over outcomes conditional on selecting each alternative~\cite{othman2010decision,chen2011information,chen2011decision}.  Since only one decision will be made and so only one conditional distribution can be sampled, simply applying a standard proper scoring rule does not result in truthful behavior.
Applying Theorem~\ref{thm:main-char} to this setting characterizes what expected scores must be, from which many of the results in these papers follow.

For example, consider the model of proper scoring rules for decision rules~\cite{chen2011information}, which we now describe:
\begin{itemize}
\item The decision maker will select an action from $\mathcal{A} = \{ 1, \ldots, n \}$;
\item After $i\in\mathcal{A}$ is chosen, an outcome from $\mathcal{O} = \{ o_1, \ldots o_m \}$ will be realized, where intuitively the probability of each outcome depends on the action chosen;
\item The decision maker asks an expert for the probabilities $P_{i,o} = \Pr[o \text{ realized} | i \text{ chosen}]$; we will denote by $\P$ the set of all allowable such probability matrices;
\item The action is chosen according to a fixed decision rule $D : \mathcal{P} \rightarrow \Delta(\mathcal{A})$ selected in advance by the decision maker, where $D_i(P)$ is the probability of choosing action $i$ given the expert reported the matrix $P$;
\item The decision maker scores the expert's report based on the chosen action $i$ and materialized outcome $o$, according to the function $\SRS_{i,o} : \mathcal{P} \to \mathbb{R} \cup \{-\infty\}$.
\end{itemize}

Given a belief $P$ and report $Q$, we can therefore write the expert's expected score as
$V(Q, P) \defeq \sum_{i, o} D_i(Q) P_{i,o} \SRS_{i,o}(Q)$.
The definition of (strict) properness for a particular decision rule then follows naturally:
a regular scoring rule $\SRS$ is {\em proper} for a decision
rule $D$ if $V(P,P) \geq V(Q,P)$ for all $P$ and all $Q \neq P$.
It
is {\em strictly proper} for the decision rule if the inequality is
strict. 

A key novelty of this setting is the \emph{type} of the agent: a matrix of conditional probabilities.  This is a different object from a distribution over $\O$ or even $\O\times\mathcal{A}$, and thus one cannot use standard scoring rule characterizations.  Similarly, the type does not describe the utility of the agent, and hence mechanism design characterizations cannot be applied.  For this reason, Chen and Kash~\citeyear{chen2011information} derive a characterization from scratch.
Fortunately, the utility of the expert in this setting is an \scorename, as the expected score $V(Q,P)$ can be written as a linear combination of the entries of $P$, and therefore is affine in $P$.
We can therefore immediately apply our theorem to derive Chen and Kash's~\shortcite{chen2011information} characterization, and even extend it to the case where the set $\P$ of probability matrices is not convex.
In the theorem statement we make use of the Frobenius inner product $P \frobeniusprod Q \defeq \sum_{i,o} P_{i,o}Q_{i,o}$.

\begin{corollary}
\label{thm:decision}
Given a set of probability matrices $\mathcal{P} \subseteq \Delta(\mathcal{O})^n$
A regular scoring rule is (strictly) proper for a decision rule $D$
if and only if
\begin{equation*}
\SRS_{i,o}(Q) = 
\begin{cases}G(Q) - G_Q \frobeniusprod Q +
\frac{G_{Q,i,o}}{D_i(Q)} & D_i(Q) > 0\\
\Pi_{i,o}(Q) 
 & D_i(Q)=0\end{cases}
 \end{equation*}
where $G: \convhull(\mathcal{P}) \rightarrow \mathbb{R} \cup \{ -\infty \}$
is a (strictly) convex function,
$G_Q$ is a subgradient of $G$ at the point $Q$
with $G_{Q,i,o} = 0$ when $D_i(Q) = 0$,
and $\Pi_{i,o} : \mathcal{P} \rightarrow R\cup
\{ -\infty \}$ is an arbitrary function that can take a value of
$-\infty$ only when $Q_{i,o}=0$.     
\end{corollary}

\begin{proof}
By Theorem~\ref{thm:main-char}, $\SRS$ is (strictly) proper for $D$ if and only if there exists a (strictly) convex $G$ such that $V(Q,P) = G(Q) + dG_Q(P-Q)$.  That is,
$$\sum_{i,o} D_i(Q) P_{i,o} \SRS_{i,o}(Q) = G(Q) - G_Q \frobeniusprod Q + \sum_{i,o} G_{Q,i,o} P_{i,o}~,$$
or for all $i$ such that $D_i(Q) \neq 0$,
$$\SRS_{i,o}(Q) = G(Q) - G_Q \frobeniusprod Q + \frac{G_{Q,i,o}}{D_i(Q)}~.$$
When $D_i(Q) = 0$, $\SRS$ is unconstrained (other than the minimal requirements regarding $-\infty$ for regularity).  However, note that our \scorename is restricted in that, because $D_i(Q)$ is fixed, some choices in $\mathcal{A}$ are not possible to select as subgradients.  In particular, it must be that $G_{Q,i,o} = 0$ when $D_i(Q) = 0$.
\end{proof}

\subsection{Proper Losses for Partial Labels}
\label{sec:partial-labels}

Several variants of proper losses have appeared in the machine learning literature, one of which is the problem of estimating the probability distribution of labels for a new data point when the training data may contain several \emph{noisy} labels, possibly not even including the correct label.  (This is frequently the case, for example, when using crowdsourced labels for training data.)   More formally, one wishes to estimate $p\in\Delta_n$ where the true label $y\in \{1,\ldots,n\}$ is drawn from $p$.  However, instead of observing a sample $y \sim p$ and designing a proper loss $\ell(\hat p,y)$, one instead only observes some noisy set of labels $S \subseteq \{1,\ldots,n\}$.  Hence, the task is to design a loss $\ell(\hat p,S)$ which when minimized over one's data yields accurate estimates of the true $p$.

Recently this problem was studied by Cid-Sueiro~\citeyear{cid-sueiro2012proper} under the assumption that the labels in $S$ are drawn i.i.d.\ from distribution $q = M p$ for some known $M\in\reals^{2^n\times n}$, where $p$ is the true label distribution.  That is, if the actual label is drawn from $p$, the noisy set of labels is drawn from $M p$ (using some indexing of the sets, say lexographical).  Cid-Sueiro provides a characterization (his Theorem 4.3) of all proper losses for an even more general version of this setting where $M$ is not known exactly but assumed to be a member of some known class; the loss should be proper for any $M$ in this class. Note that the (negative) payoff $\E_{S\sim Mp}[\ell(\hat p, S)] = \ell(\hat p, \cdot)^\tr M p$ is linear in the underlying distribution $p$, so our Theorem~\ref{thm:main-char} applies and allows us to recover his characterization result.  Note that this is essentially a latent observation setting, and the fact that what we observe is a set of labels is in no way necessary; any observed outcome whose distribution has a linear (or affine) relationship with the latent outcome would suffice to apply our theorem.

Rather than introducing the full, general model used by Cid-Sueiro, we show how our theorem applies to yield a characterization for a single, fixed $M$.  This is a special case of his Theorem 4.3, generalized to allow restricted sets of probability distributions $\P$.
\begin{corollary}
  \label{thm:partial-label-char}
  Let number of labels $n$, matrix $M \in\reals^{2^n\times n}$ , and $\P \subseteq \Delta_n$ be given such that $M \P$ has full dimension in the column space of $M$.
  Let outcome set $\O = 2^{\{1,\ldots,n\}}$ be the power set of $\{1,\ldots,n\}$.
  A regular score $\SRS:\P \times \O \to\extreals$ is proper if and only if there exists a convex function $G:\convhull(\P)\to\reals$ such that
  \begin{equation}
    \label{eqn:sr-pl}
    \ell(\hat p,S) =  - G(\hat p) - G_{\hat p}^\tr (M^+e_S) + G_{\hat p}^\tr \hat p,
  \end{equation}
  where $e_S$ is a one-hot encoding (indicator vector) of $S$, $M^+$ is a left inverse of $M$,  and $G_p$ is a subgradient of $G$ for all $p\in\P$.
\end{corollary}
\begin{proof}
For the form given, the expected loss is
\begin{align*}
\E_{S\sim Mp}[\ell(\hat p,S)
&=  \E_{S\sim MP} [- G(\hat p) - G_{\hat p}^\tr (M^+e_S) + G_{\hat p}^\tr \hat p]\\
&=  - G(\hat p) - G_{\hat p}^\tr M^+M p + G_{\hat p}\tr \hat p\\
&=  - G(\hat p) - G_{\hat p} \cdot (p - \hat p)~.
\end{align*}
As observed above take $\A = \{p \mapsto  -\E_{S\sim Mp} \ell(q,S) \mid q \in \P\}$, each of which is linear in $p$.
  Propriety of the given form then follows immediately from Theorem~\ref{thm:main-char}.
  For the converse, let $\ASdef$ be a given truthful \scorename for $\T = \P$.
  From the theorem, $\AS{p,\cdot} = G(p) + d_{p}(\,\cdot\, - p) \in \A$.
  Thus we can write $\ell(\hat p, \cdot)^\tr M p = \E_{S\sim Mp} [\ell(\hat p,S)] = - G(\hat p) - G_{\hat p}\cdot (p - \hat p)$.
  Taking $\hat \ell(\hat p, \cdot) =  \ell(\hat p, \cdot) + (G(\hat p) - G_{\hat p} \cdot \hat p)\ones$, this means $\hat \ell(\hat p, \cdot)^\tr M p = - G_{\hat p}\cdot p$ for all $p \in \P$.  By our assumption that $M \P$ has full dimension in the column space of $M$,  $\hat \ell(\hat p, \cdot)^\tr = - M^+ G_{\hat p}$ for some left inverse $M^+$, showing $\ell$ is of the desired form.  (If $M \P$ does not have full dimension there may be additional choices of $\ell$ with the correct expected value on $\P$, but the proof otherwise applies.)
\end{proof}

\subsection{Mechanism Design with Partial Allocation}

Several mechanism design settings considered in the literature have some form of \emph{exogenous} randomization, in that ``nature'' chooses some outcome $\omega$ according to some (often unknown) distribution, and this distribution may depend on the allocation chosen by the mechanism.
Examples include sponsored search auctions~\cite{feldman2008algorithmic}, multi-armed bandit mechanisms~\cite{babaioff2009characterizing}, and recent work on daily deals~\cite{cai2013designing}.
The work of Cai, Mahdian, Mehta, and Waggoner~\shortcite{cai2013designing} introduces a very general model for such settings, which will be our focus in this subsection.

In the setting of Cai, et al., the mechanism designer wants to elicit two pieces of information: the agent's (expected) value for an item in an auction and the probability distribution of a random variable conditional on that agent winning.
Their goal is to understand how the organizer of a daily deal site can take into account the value that will be created for users (as opposed to just the advertiser) when a particular deal is chosen to be advertised.
(E.g. the site operator may prefer deals that sell to many users over equally profitable deals that sell only to a few because this keeps users interested for future days.)
The authors characterize the possible (implementable) ways of quantifying user welfare as a function of the agent's probabilistic belief as to the outcome (of nature) resulting from each allocation.
We will show how to recover this characterization as a special case of a more general setting in which a mechanism designer wishes to elicit two pieces of information, but the second need not be restricted to probability distributions.

We begin with a description of the setting of Cai et al.~\shortcite{cai2013designing}.
Let $\O$ be a set of outcomes, and for each outcome $\o$ and each agent $i$, let $\Omega^{i,\o}$ be some set of events.
For example, $\o$ could determine which agent wins an auction for the opportunity to advertise a special offer from its business and $\Omega^{i,\o}$ could represent the set of numbers of customers that may purchase the deal.
Agents each have a valuation function $v^i:\O\to\reals$ and a set of beliefs $p^{i,\o} \in \Delta(\Omega^{i,\o})$ for each allocation $\o\in\O$, e.g., an expected value for getting to advertise and a probability distribution over the number of customers who accept the deal.
The mechanism aggregates all of this information into a single outcome $\o$, and additionally choses some payoff function $s^i : \Omega^{i,\o} \to \reals$, so that the final utility of agent $i$ is $v^i(\o) + \E_{p^{i,\o}}[s^i]$; that is, the winning agent both gets to advertise and accepts a scoring rule contract regarding its prediction of the number of customers.
A mechanism is truthful if for all values of $v$ and $p$ for the other agents, agent $i$ maximizes her total utility by reporting $v^i$ and $p^i \defeq (p^{i,\o})_{\o\in\O}$ truthfully.
For additional examples, the standard sponsored search setting has $\Omega^{i,\o} = \{\text{click},\text{no click}\}$ for $\o$ such that $i$ is allocated a slot, and the probabilities $p^{i,\o}$ are assumed to be public knowledge.
Moreover, the decision rules framework discussed above is a single-agent special case with $v \equiv 0$ and $\Omega^\o = \Omega^{\o'} = \Omega$ for all $\o\in\O$; of course, unlike the interpretation above, in this setting $\o\in\O$ is the allocation/decision while $\Omega$ is the set of outcomes.

Motivated by incorporating the utilities of the end consumers in a daily deal setting, Cai et al.~\shortcite{cai2013designing} ask when one can implement an allocation rule of the form $f(v,p) = \argmax_{\o\in\O} v(\o) + g^\o(p^\o)$, which they interpret as maximizing welfare of the winner plus a term that captures something about the welfare of consumers.
In other words, when does there exist some choice of payment making $f$ truthful.
The authors conclude that this can be done if and only if $g^\o$ is convex for each $\o\in\O$.
In what follows, we will recover this result using our \scorename framework.

We first observe that this model can easily be cast as an \scorename, as follows.
For simplicity, we fix some agent $i$ and focus on the single-agent case; as discussed several times above, this is essentially without loss of generality.
The type space is simply (as subset of) the combined private information of the agent,
\begin{equation}
\label{eq:general-model-1}
\T \subseteq \left\{(v,p) : v \in \O\to\reals,\; p \in \prod_{\o\in\O} \Delta(\Omega^{i,\o}) \right\}.
\end{equation}
By assumption, the utility of the agent upon allocation $\o$ and payoff $s$ is simply 
$v(\o) + \E_{p^\o}[s]$,
 which is linear in the type $t=(v,p)$ and therefore affine.
Thus the net payoff
$\AS{t',t} = v(\o(t')) + \E_{p^{\o(t')}}[s(t')]$
is an \scorename, with score set $\A = \{(v,p) \mapsto v(\o) + \E_{p^\o}[s] : \o\in\O, s\in \Omega^{i,\o}\to\reals\}$, where again we have fixed $i$.

To answer the implementability question of Cai, et al.~\cite{cai2013designing}, we consider a general type space of the form $\T \subseteq  \V = \V^X \times \V^Y$ which can be expressed as a partition into two (subsets of) subspaces.
In the daily deal setting, we would have $\V^X = \reals^\O$ and $\V^Y = \prod_{\o\in\O} \Delta(\Omega^{i,\o})$, but in general $\V^Y$ need not be restricted to probability distributions.

We wish to know when a function $f^X:\T\to\linear(\V^X\to\reals)$ is \emph{extendable}, in the sense that there exists some truthful affine score $\ASdef$ with score set  $\A\subseteq\affine(\V\to\reals)$, and some $f^Y:\T\to\affine(\V^Y\to\reals)$ such that $\AS{t',t} = f^X(t')(t^X) + f^Y(t')(t^Y)$, where of course $t = (t^X,t^Y)$.
In the daily deals context, $f^X$ selects the outcome $\o$, formally represented as the linear map $v \mapsto v(\o)$, and $f^Y$ is an expected score of the form $\E_{\omega\sim p^\o}[s]$, which is affine (here linear) in $p$.
Thus, asking whether $f^X$ is extendable is equivalent to asking which rules for selecting the outcome $\o$ are implementable.

To answer this question, let us introduce some notation.
For each $a \in \A$ we write $X(a) \in \linear(\V^X\to\reals)$ to be the linear part of $a$ on $\V^X$, and $Y(a)$ to be the \emph{affine} part of $a$ on $\V^Y$.
We will write $X(\A) \defeq \{X(a) : a\in\A\}$.
In this very general framework, we can show the following.

\begin{theorem}[Informal]
  Partial allocation rule $f^X$ is extendable if and only if
  \begin{equation}
    \label{eq:general-model-3}
    f^X(t) \in \argsup_{x \in X(\A)} \left\{ x(t^X) + \sup_{\substack{a\in\A\\X(a) = x}} \left\{ Y(a)(t^Y) \right\} \right\}~.
  \end{equation}
\end{theorem}
To see this, note that one direction follows from the fact that an \scorename is truthful if and only if $\ASS(t) \in \argsup \left\{ a(t) : {a\in\A} \right\}$; we simply take the supremum first over $X(\A)$ and then over the rest.
For the other direction, note that taking $\AS{t',t} = f(t')(t^X) + y(t')(t^Y)$ where $y$ is in the $\argsup$ of the supremum of eq.~\eqref{eq:general-model-3} gives a truthful \scorename.

Returning to the special case of daily deals, let us denote by $a_{\o,s}\in\A$ the function $(v,p)\mapsto v(\o)+\E_{p^{\o}}[s]$.  We now see that $f(v,p)$ is implementable if and only if it satisfies
\begin{equation}
  \label{eq:general-model-4}
  f(v,p) \in \argsup_{\o\in \O} \left\{ v(\o) + \sup_{s \,:\, a_{\o,s} \in \A} \{ \E_{p^\o}[s] \} \right\}.
\end{equation}
Thus, letting $g^\o(p^\o) = \sup \left\{ \E_{p^\o}[s] : a_{\o,s} \in \A\right\}$, we see that $g^\o$ is convex as the supremum of affine functions.  Moreover, given any collection of convex functions $\{g^\o\}_{\o\in\O}$, where $g^\o:\Delta(\Omega^{i,\o})\to\reals$, we can define $S^\o \defeq \{\omega \mapsto g(p) + dg(\ones_\omega - p) : p\in\dom(g)\}$ and $\A \defeq \{a_{\o,s} : \o\in\O,\,s\in S^\o\}$, thus recovering each $g^\o$ in the above expression.  It then only remains to show that no other nonconvex function can serve in the $\argsup$; for this one may appeal to the argument of Cai et al.~\shortcite{cai2013designing} which observes that the indifference points between different allocations are fixed, thus determining the function in the $\argsup$ up to a constant.

A special case of the above, but closer to classical mechanism design, is captured in the following scenario.
The mechanism designer has two distinct sets of goods to allocate and wants to design a truthful mechanism that is consistent with a partial allocation rule that determines how the primary goods should be allocated given the agent's preferences over both types of goods.
Such mechanisms are characterized by the following informal theorem.

\begin{theorem}[Informal]
\label{thm:partial-informal}
Consider an agent with type $t = (t_1,t_2)$.  A truthful \scorename $\ASdef$ with score set $\A$ (which we represent as a set of triples $(o_1,o_2,p)$ of first outcome, second outcome, and price) is consistent with a partial allocation rule $f : t \mapsto o_1$, if and only if
\begin{equation}
  \label{eq:general-model-3-informal}
  f(t) \in \argsup_{o_1'} \left\{ t_1(o_1') + \sup_{\substack{(o_1,o_2,p) \in \A\\o_1' = o_1}} \left\{ t_2(o_2)+p \right\} \right\}
\end{equation}
\end{theorem}
In particular, analogous to the daily deals setting, the mechanism designer is restricted to mechanisms that make decisions based on a convex function of $t_2$ (the inner supremum is a pointwise supremum over affine functions and thus convex).

We conclude by noting similarity to recent work of Chambers and Lambert~\citeyear{chambers2014dynamically}, where a center wishes to elicit an agent's belief about some future event, and to do so multiple times as the event gets closer.  The solution proposed is essentially to have the agent choose a menu of scoring functions at time $0$ from a menu of menus, and then at time $1$ choose a score from the menu chosen at time $0$.  Both the present setting and theirs share this sense of ``menu of menus'', however one can easily check that the two are incomparable.  In particular, the relationship between the final score and the beliefs of the agent can be nonlinear in the Chambers--Lambert model.

\subsection{Responsive Lotteries}
\label{sec:affine-responsive-lotteries}

Feige and Tennenholtz~\shortcite{feige2010responsive} study the problem of how an agent can be incentivized to indirectly reveal his utility function over outcomes by being given a choice of lotteries over those outcomes, an approach with applications to experimental psychology, market research, and multiagent mechanism design.
Given a finite set of possible outcomes, the authors give several examples of effective lotteries under the assumption of risk neutrality.
In contrast, our approach allows us to give a complete characterization, which highlights the relationship between natural desiderata and underlying geometric properties of the set of possible lotteries, as we now describe.

The authors ask when a responsive lottery is \emph{truthful dominant}, which is defined as having the following three properties: \emph{incentive compatibility}, meaning the true utility is among the optimal reports, \emph{rational uniqueness}, meaning the optimal lottery for a given belief is unique, \emph{rational invertibility}, meaning every report can be optimal for at most one utility.
We would like to relate these notions to simple geometric properties, which we introduce now informally, and formalize later in Definition~\ref{def:strictly-convex-smooth}.
A convex set $K$ is \emph{strictly convex} if no point on its boundary can be expressed as a convex combination of other points in $K$, and $K$ is \emph{smooth} if each point on its boundary has a unique unit normal vector.
Using the results to follow, we can show that strict truthfulness and continuity of the lottery rule jointly correspond to strict convexity of the lottery set, and uniqueness of the utility given the optimal lottery corresponds to smoothness of the boundary.

\begin{corollary}
  \label{cor:affine-lotteries}
  A lottery rule $f$ satisfies incentive compatibility and rational uniqueness if and only if $f(x) = \argmax_{p\in K} \inprod{x,p}$ for $K\subset \Delta_n$ compact and strictly convex relative to $\Delta_n$.
  Moreover, $f$ additionally satisfies rational invertibility (and thus is truthful dominant) if and only if $K$ is additionally smooth.
\end{corollary}
\begin{proof}
  In this setting, we can see that utility functions are equivalent up to positive-affine transformations.
  (Since there are no payments, multiplying the utility of each outcome by a positive constant or adding a constant to the utility for each outcome has no effect on the optimal lottery for an agent.)
  Thus, we may project the utilities and probability simplex onto the set $V = \{x\in\reals^n : \sum_i x_i = 0\}$, which only changes the expected utilities by a constant.
  We then write these vectors in a basis for $V \cong \reals^{n-1}$, normalize the utilities (only scaling them) to the unit sphere in $V$, and apply Theorem~\ref{thm:affine-linear}.
\end{proof}

We now turn to the geometric statements needed to establish Corollary~\ref{cor:affine-lotteries},
beginning with some formal definitions.
We denote by $\partial K$ the boundary of the set $K \subseteq \reals^n$.

\begin{definition}
  \label{def:exposed-normal}
  Given a compact convex set $K \subset \reals^n$, we define the \emph{exposed face} $F_K(t)$ in direction $t\neq 0$ and the \emph{normal cone} $N_K(k)$ at point $k\in\partial K$ by
  \begin{equation}
    \label{eq:affine-1}
    F_K(t) = \argmax_{k\in K} \inprod{t,k}, \qquad N_K(k) = \{ t\in\reals^n : k \in F_K(t) \}.
  \end{equation}
\end{definition}

An exposed face $F_K(t)$ is simply the set of points as far in direction $t$ as possible; for example, on a triangle $ABC$ the exposed faces are vertices $A$,$B$,$C$ and the edges $\overline{AB},\overline{BC},\overline{CA}$.  The normal cone $N_K(k)$ is simply the set of all valid normal vectors to $K$ at a point $k$ on the boundary of $K$; for the triangle, the normal cone at a point on an edge is simply a ray perpendicular to it, but on a vertex there is a closed cone spanning the exterior angle of the triangle.

\begin{definition}
  \label{def:strictly-convex-smooth}
  We say $K$ is \emph{strictly convex} if $F_K(t)$ is a singleton for all $t\neq 0$.  Dually, we say $K$ is \emph{smooth} if $N_K(k)$ is a ray (i.e. $\{\alpha t : \alpha \geq 0\}$ for some $t\neq 0$) for all $k\in\partial K$.
\end{definition}

To illustrate: a disc is smooth and strictly convex, a triangle is neither, a rounded rectangle is smooth but not strictly convex, and the intersection of two discs is strictly convex but not smooth.  We now establish the tight connection between these geometric concepts and the truthfulness conditions in this setting, which will imply Corollary~\ref{cor:affine-lotteries}.

\begin{theorem}
  \label{thm:affine-linear}
  Let $\T = \{t\in\reals^n : \|t\|_2=1\}$ be the unit sphere in $\reals^n$, and let a truthful affine score $\ASdef$ with score set $\A \subseteq \linear(\reals^n\to\reals) \cong \reals^n$ be given.  Then $S(t) \defeq \AS{t,\cdot}$ is surjective and continuous (as a function to $\reals^n$) and $\ASS$ is strictly truthful if and only if $\A$ is the boundary of a compact and strictly convex set $K\subset\reals^n$.  $S$ is additionally injective if and only if $K$ is additionally smooth.
\end{theorem}

\begin{proof}
  We begin with the first part of the theorem.  Let $K$ be compact and strictly convex, and $\A = \partial K$.  Then as $\ASS$ is truthful, we must have $S(t) \in \argsup_{a\in \A} \inprod{t,a}$.  As $\A = \partial K$, we may also write $S(t) \in \argmax_{k\in K} \inprod{t,k}$.  Now by strict convexity of $K$, we have for every $a\in\A = \partial K$, there exists some $t\in\T$ such that $\{a\} = \argmax_{k\in K} \inprod{t,k}$, giving us both surjectivity and strict truthfulness (as $S(t) = a$).  Continuity follows immediately from Berge's Maximum Theorem~\cite{ok2007real}.

  For the converse, let $S$ be strictly truthful, surjective, and continuous.  By standard arguments, since $\T$ is a compact subset of $\reals^n$, we have $\A = S(\T)$ is compact as a continuous image of a compact set.  Thus, $K\defeq \conv(\A)$ is a compact convex set.
  Letting $F_K(t) \defeq \argmax_{k\in K} \inprod{t,k}$ be the exposed face of $K$ in direction $t$, we will now show $F_K(t) = \{S(t)\}$.
  First, observe that the extreme points of $K$, $\ext(K)$, are a subset of $\A$; otherwise we have $k \in \ext(K)\setminus \A$, so $K\setminus\{k\}$ is a convex set containing $\A$, contradicting the definition of $K = \conv(\A)$.
  Now we may apply~\cite[Proposition A.2.4.6]{urruty2001fundamentals} to express the $\argmax$ in terms of the extreme points of $K$, giving us
\begin{equation*}
  F_K(t) \defeq \argmax_{k\in K} \inprod{t,k} = \conv\left(\argmax_{k\in \ext(K)} \inprod{t,k}\right) \subseteq \conv\left(\argmax_{a\in \A} \inprod{t,a}\right) = \{S(t)\}.
\end{equation*}
As $K$ is compact, $F_K(t)$ is nonempty, and thus $F_K(t) = \{S(t)\}$, and additionally we conclude $S(t) \in \ext(K)$.  Hence $\A = S(\T) \subseteq \ext(K)$, and as we concluded the reverse conclusion above, we have $\A = \ext(K)$.  We now apply~\cite[Proposition C.3.1.5]{urruty2001fundamentals} to obtain $\partial K = \bigcup_{t\in \T} F_K(t)$, which in turn gives $\partial K = \A$ by surjectivity.  Finally, as $\ext(K) = \A = \partial K$, we have strict convexity of $K$.

For the final statement of the theorem, we note that by~\cite[Proposition C.3.1.4]{urruty2001fundamentals}, we have $k \in F_K(t) \iff t \in N_K(k)$.  By the above, we already have $F_K(t) = \{S(t)\}$ for all $t\in \T$, which implies $N_K(k)\cap\T = \{t : S(t)=k\}$.  Hence, $N_K(a)$ is a ray for all $a\in\A$ if and only if $S$ is injective.
\end{proof}

\subsection{Crowdsourcing}

A common technique to provide incentives in crowdsourcing is to ask agents to answer a series of questions, the answers to which a small fraction are known, and then score agents based on how well their responses match the known answers.
A classic problem, however, is to design mechanisms which incentivize agents to tell the truth but sufficiently deter ``spammers'' who may wish to submit low-quality or even useless information in return for small but nonzero rewards.
Witkowski et al.~\citeyear{witkowski2013dwelling} were perhaps the first to address this problem, by incentivizing agents to ``pass'' on a question if they do not know the answer with sufficient confidence.
(In practice, only a small subset of the questions have known answers, and the incentives for these questions are leveraged to obtain worker responses to the other questions.  Following standard arguments, if the worker believes this subset to be chosen randomly, as can be achieved by a random permutation, then the incentives are the same as if answers were known for every question.)

Similarly, Shah and Zhou~\citeyear{shah2015double} introduce a mechanism which satisfies a simple axiom: if every attempted question is answered incorrectly (among those whose answers are known), the agent should receive 0 payoff.
They then show that their mechanism, which is based on thresholds, is the only such mechanism satisfying this axiom.
The authors also extend the ideas of Witkowski et al.~\cite{witkowski2013dwelling} to allow for multiple ``confidence'' levels, and define an attempt to be any answer with non-zero confidence.

The proofs in Shah and Zhou~\citeyear{shah2015double} are quite involved, spanning over four pages.
Using our framework, we can recover this characterization result in a few paragraphs, as we now demonstrate.
The crux of the argument is in translating the axiom to geometric constraints on the mechanism and then applying general results about the geometry of mechanisms with a finite set of allocations.  While we apply these results in the proof here, we will discuss them in Section~\ref{sec:props-finite}.

To begin, let us introduce the setting of Shah and Zhou~\citeyear{shah2015double} in terms of affine scores.
We will take the type of the worker to be a vector $t = (p_1,\ldots,p_n) \in [0,1]^n$ which encodes the worker's probability of correctly answering each question if it is not skipped, with the outcome being an $n$-dimensional binary vector $o\in\{-1,+1\}^n$ encoding for each question whether the agent answered correctly ($+1$) or not ($-1$).  The ``report'' of the worker in this formulation would then be to answer or skip, denoted $r\in\{-1,0,+1\}^n$ where $0$ denotes skipping.  Note that in the case of a skip, the outcome is both meaningless and irrelevant.  (Of course, in practice a worker would report an actual answer, but the analysis simplifies to the above by abstracting away the separate choices for each answer.)  
We can think of a payment function $f$ then as taking whether the agent reported and whether the agent would have been correct had she reported and computing a payment based on this; formally, $f:\{0,1\}^n \times \{-1,+1\}^n \to \reals$.  To capture the fact that the designer does not actually know whether the agent was right or wrong when skipping an answer, we require that $f((0,r_{-i}),(-1,o_{-i})) = f((0,r_{-i}),(+1,o_{-i}))$. Note that while $f$ takes whether the agent reported rather than what the report was, in a slight abuse of notation we use $r$ for both.  The expected payment to the worker is $\E_{o_i\sim p_i} f(r,o) = f(r,\cdot) \cdot t$, which is linear in $t$ and therefore an affine score.  Given a confidence threshold parameter $T$, the payment function is deemed \emph{incentive-compatible (IC)} if for each question $i$, the worker is incentivized to choose $r_i=0$ (skip) if $p_i < T$, and to choose $r_i=1$ (answer the question) if $p_i > T$.

We now show how to use our framework to rederive the core results of Shah and Zhou~\citeyear{shah2015double}.  Let us first restate their no-free-lunch axiom and mechanism in our framework.

\begin{axiom}[No-free-lunch]
  If all the answers attempted (not skipped) by the worker are wrong, then the payment is zero.
\end{axiom}

Define the following \emph{multiplicative mechanism} given a threshold $T\in(0,1)$ and budget $\mu>0$ which must be fully spent if all questions are answered correctly: 
\begin{equation}
  \label{eq:double-or-nothing-mechanism}
  f(r,o) = \mu T^{n-C} \ones\{W=0\}~,
\end{equation}
where $C = |\{i\in\{1,\ldots,n\}:r_i=1,o_i=1\}|$ is the number of correct answers, and $W = |\{i\in\{1,\ldots,n\}:r_i=1,o_i=-1\}| = \|r\|_1 - C$ is the number of incorrect answers (among those not skipped).  The name ``multiplicative'' comes from viewing $f$ as the product of scores for each question: $0$ if answered incorrectly, $1/T$ if answered correctly (recall $T<1$), and $1$ if skipped.

\begin{theorem}
  \label{thm:double-or-nothing-sufficient}
  The multiplicative mechanism~\eqref{eq:double-or-nothing-mechanism} is IC and satisfies the no-free-lunch axiom.
\end{theorem}

\begin{proof}
  The axiom is clearly satisfied by the $\ones\{W=0\}$ term.  For IC, consider the ``property'' $\Gamma:[0,1]^n\to\{0,1\}^n$ given by $\Gamma(t) = \max_{r\in\{0,1\}^n} f(r,\cdot)\cdot t$.  By definition, the multiplicative mechanism elicits $\Gamma$.  Combining Corollary~\ref{cor:wlt} and the results on finite properties in Section~\ref{sec:props-finite}, we see that it suffices to check IC at the boundary of each level set of $\Gamma$, corresponding to the types that are indifferent between reporting and skipping for at least one question.  This we can check directly, because answering the question ($r_i=1$) multiplies the payment by $1/T$ with probability $T$ and makes it 0 with probability $1-T$, a worker is indeed indifferent when $p_i=T$ as desired.
  
Verifying local truthfulness also requires checking types in the the interiors of the level sets but close to the boundary.  It follows easily from the indifference of types on the boundary that these types strictly prefer their assigned allocation.
\end{proof}

\begin{theorem}
  \label{thm:double-or-nothing-necessary}
  The multiplicative mechanism is the only IC mechanism satisfying the no-free-lunch axiom.
\end{theorem}

\begin{proof}
Let $f$ be any IC payment function satisfying the no-free-lunch axiom.  Again consider the finite property $\Gamma$ from Theorem~\ref{thm:double-or-nothing-sufficient}, which by Theorem~\ref{thm:power-diag} is a power diagram.  From the proof of Theorem~\ref{thm:power-diag}, without loss of generality we may take the sites of the corresponding power diagram to be the subgradients of the corresponding ``consumer surplus'' function given by $G(t) = \max_{r\in\{0,1\}^n} f(r,\cdot)\cdot t$, which corresponds to the maximum expected payment a worker can receive given their type $t$.

From the constraint $f((0,r_{-i}),(-1,o_{-i})) = f((0,r_{-i}),(+1,o_{-i}))$ that the score for skipping does not depend on the true answer, the subgradient for all skipped must be all zeroes.

Now observe that by the structure of the indifference sets where $p_i=T$ for exactly one $i$, and the fact that the line segment connecting sites in adjacent cells must be perpendicular to their boundary, we can conclude that the sites must lie on the vertices of an $n$-orthogon (generalization of a rectangle to $n$ dimensions).  To see this, consider that for any $r,r'\in\{0,1\}^n$ differing in entry $i$, the line segment connecting the corresponding sites $s_r$ and $s_{r'}$ must be perpendicular to the (affine) hyperplane $\{t:t_i=T\}$.  Piecing these constraints together shows that fixing the site for $r^0=(0,\ldots,0)$ and each $r^i$, which is zero in every entry but $r_i=1$, uniquely determines the corresponding $n$-orthogon on which the sites reside as vertices.  Moreover, each of these sites is non-zero only in the dimensions corresponding to the questions answered.  By the no-free-lunch axiom, all of the weights of the power diagram are 0.  (Note that this also shows the no-free-lunch axiom could be (apparently) further weakened to only apply when a single question is answered.) The requirement to pay the budget $\mu$ if all questions are answered correctly then fixes the scaling. 
\end{proof}

\section{Property Elicitation}
\label{sec:properties}

We wish to generalize the notion of truthful elicitation from eliciting private information from some set $\T$  to accept reports from a space $\R$ which is different from $\T$.  To even discuss truthfulness in this setting, we need a notion of a truthful report $r$ for a given type $t$.  We encapsulate this notion by a general multivalued map which specifies all (and only) the correct values for $t$.

\begin{definition}
Let  $\T$ be a given type space, where $\T\subseteq\V$ for some vector space $\V$ over $\reals$, and $\R$ be some given report space.  A \emph{property} is a multivalued map $\Gamma : \T \toto \R$ which associates a nonempty set of correct report values to each type.  We let $\Gamma_r \defeq \{t\in \T \,|\, r\in\Gamma(t)\}$ denote the set of types $t$ corresponding to report value $r$.
\end{definition}
One can think of $\Gamma_r$ as the ``level set'' of $\Gamma$ corresponding to value $r$.  This concept will be especially useful when we consider finite-valued properties in Section~\ref{sec:props-finite}.  A natural constraint to impose on these level sets is that they be \emph{non-redundant}, meaning no property value $r$ has a level set entirely contained in another.
\begin{definition} \label{def:props-redundant}
  Property $\Gamma:\T\toto \R$ is \emph{redundant} if there exist $r,r'\in \R$ such that $\Gamma_{r'}\subseteq \Gamma_{r}$.  Otherwise, $\Gamma$ is \emph{non-redundant}.
\end{definition}
The non-redundancy condition is essentially a bookkeeping tool.  If one adds report elements $r'$ which are (weakly) dominated by another report $r$, then any time $r'$ would be correct, an agent could safely report $r$ instead.  Hence, one could think of imposing this condition then as simply ``pre-processing'' $\Gamma$ to remove any dominated reports.

We extend the notion of an \scorename to this setting, where the report space is $\R$ instead of $\T$ itself.  Note that the score set $\A = \{\AS{r,\cdot} ~|~ r \in \R \}$ is still a subset of $\affine(\V\to\extreals)$.
\begin{definition}
  An \scorename $\ASpropdef$ \emph{elicits} a property $\Gamma:\T\toto \R$ if for all $t$,
  \begin{equation}
    \label{eq:gamma-truthful}
    \Gamma(t) = \argsup_{r\in \R} \AS{r,t}.
  \end{equation}
  If we merely have $\Gamma(t) \subseteq \argsup_{r\in \R} \AS{r,t}$, we say \ASS\ \emph{weakly elicits} $\Gamma$.
  A property $\Gamma:\T\toto \R$ is \emph{elicitable} if there exists some \scorename $\ASpropdef$ eliciting $\Gamma$.
\end{definition}
Note that it is certainly possible to write down $\ASS$ such the $\argsup$ in~\eqref{eq:gamma-truthful} is not well defined.  This corresponds to some types not having an optimal report, which we view as violating a minimal requirement for a sensible affine score.  Thus, as we have defined properties to be nonempty, in order for \ASS\ to be an \scorename, we require \eqref{eq:gamma-truthful} to be nonempty for all $t \in \T$.

As before, we allow $\AS{r,t}$ to take on values in the extended reals to capture scoring rules such as the log score, so we need a notion of regularity --- an \scorename \ASS\ is \textsl{$\Gamma$-regular} if $\AS{r,t} < \infty$ always and $\AS{r,t} \in \reals$ whenever $r\in\Gamma(t)$.  We define $\Gamma$-regular linear and affine families similarly.\footnote{The family $\{ \ell_r \in \linear(\V\to\extreals) \}_{r  \in \R}$ is \emph{$\Gamma$-regular} if $\ell_r(t) \in \reals$ for all $t \in \Gamma_r$, and $\ell_{r}(t') \in \reals \cup \{-\infty\}$ for $t' \in \T \setminus \Gamma_r$.  Likewise for $\Gamma$-regular affine functions.}

\subsection{A First Characterization}

The simplest way to come up with an elicitable property is to induce one from an \scorename.  For any $\ASpropdef$ with score set $\A \subset \affine(\V\to\extreals)$, the property
\begin{equation} \label{eq:properties-3}
  \Gamma^\ASS : t \to \argsup_{r\in \R} \AS{r,t}
\end{equation}
is trivially elicited by \ASS\ if this $\argsup$ is well defined.

Observe also that any \scorename \ASS\ eliciting $\Gamma$ gives rise to a truthful \scorename in the original sense --- in fact, this is a version of the \emph{revelation principle} from mechanism design.  For each $t$ let $r_t \in \Gamma(t)$ be a report choice for $t$; then the \scorename $\AS[^\T]{t',t} \defeq \AS{r_{t'},t}$ is truthful.  Moreover, by our choices of $\{r_t\}$, we have
\begin{equation}
  \label{eq:revelation-g}
  G(t)\defeq \sup_{t'\in\T} \AS[^\T]{t',t} = \sup_{r\in \R} \AS{r,t}.
\end{equation}
Of course, in general, $\ASS^\T$ will not be strictly truthful, since by definition, any reports $t',t''$ with $r_{t'} = r_{t''}$ will have $\AS[^\T]{t',\cdot} \equiv \AS[^\T]{t'',\cdot}$.  Thus we may think of a property as \emph{refining} the notion of strictness for a truthful \scorename.  The connection we draw in Theorem~\ref{thm:prop-char} is that, in light of~\eqref{eq:revelation-g}, a property $\Gamma$ therefore specifies the portions of the domain of $\T$ where $G$ must be ``flat''.
To get at the connection between properties and ``flatness'', we start with a technical lemma which shows that having the same subgradient at two different points is equivalent to $G$ being flat in between.

\begin{lemma}\label{lem:subgradients}
  Let $G:\convhull(\T)\to\extreals$ be convex with $G(\T)\subseteq \reals$, and let $d\in \partial G_t$ for some $t\in\T$.  Then for all $t'\in \T$,
\[ d \in \partial G_{t'} \; \iff \; G(t) - G(t') = d(t-t'). \]
\end{lemma}
\begin{proof}
  First, the forward direction.  Applying the subgradient inequality~\eqref{eq:subgradient} at $t'$ for $dG_t = d$ and at $t$ for $dG_{t'} = d$, we have
  \begin{align*}
    G(t') &\geq G(t) + d(t'-t) \\
    G(t) &\geq G(t') + d(t-t'),
  \end{align*}
  from which the result follows (as $G(t)$ and $G(t')$ are finite).

For the converse, assume $G(t) = G(t') + d(t-t')$ and let $t''\in\T$ be arbitrary.  As $G(t)$ and $G(t')$ are finite, $d(t-t')$ is as well.  Then using the subgradient inequality~\eqref{eq:subgradient},
  \begin{align*}
    G(t') + d(t''-t') &= G(t') + d(t''-t) + d(t-t')
    = G(t) + d(t''-t) 
    \leq G(t''). \qedhere
  \end{align*}
\end{proof}

We are now ready to state our first characterization, which in essence says that eliciting a property $\Gamma$ is equivalent to eliciting subgradients of a convex function $G$.  Intuitively, by truthfulness the linear part of $\AS{r,\cdot}$ must be a subgradient of $G$ at all $t \in \Gamma_r$.  The lemma shows that this equivalent to flatness, which means we can calculate $G$ on $\Gamma_r$ set by picking any type $t_r \in \Gamma_r$ and following the subgradient.  Since all choices of $t_r$ lead to the same value, we could just as easily ask for this subgradient $\varphi(r)$ to be reported directly.
As subgradients are functions (in this case from $\T$ to $\extreals$), we use the curried notation $\varphi(r)(t)$ for the application of this function.
Recall that we overload the $\partial G$ notation to refer to the full set of possible subgradients $\partial G = \cup_ {t \in \T} \partial G_t$.

\begin{theorem}\label{thm:prop-char}
  Let non-redundant property $\Gamma:\T\toto \R$ and $\Gamma$-regular \scorename $\ASpropdef$ be given.  Then \ASS\ elicits $\Gamma$ if and only if there exists some convex $G : \convhull(\T)\to\extreals$ with $G(\T)\subseteq\reals$, some $\D\subseteq \partial G$, some bijection $\varphi : \R \to \D$ with $\Gamma(t) = \varphi^{-1}(\D\cap \partial G_t)$, and some $\{t_r\}_{r\in \R} \subseteq \T$ satisfying $r'\in\Gamma(t_{r'})$ for all $r'$, such that for all $r\in \R$ and $t\in\T$,
  \begin{equation}
    \label{eq:prop-char}
    \AS{r,t} \; = \; G(t_r) + \varphi(r)(t-t_r).
  \end{equation}
Moreover, if \ASS\ elicits $\Gamma$ it can be written in this form for any $\{t_r\}_{r\in \R} \subseteq \T$ satisfying $r'\in\Gamma(t_{r'})$ for all $r'$.
\end{theorem}
\begin{proof}
For the converse, let $\ASS$ be given of the form~\eqref{eq:prop-char}.  We show that it elicits $\Gamma$, i.e. $\Gamma(t) = \argsup_{r\in \R} \AS{r,t}$.
The third line of the derivation applies Lemma~\ref{lem:subgradients}.
\begin{align*}
r \in \Gamma(t)
&\iff r \in \varphi^{-1}(\D\cap \partial G_t)\\
&\iff \varphi(r) \in \D \cap \partial G_t\\
&\iff \AS{r,t} = G(t)\\
&\iff r \in \argsup_{r'\in \R} \AS{r',t}~.
\end{align*}
Note that any choice of $\{t_{r'}\}_{r'\in \R}$ such that $r'\in\Gamma(t_{r'})$ for all $r'$ would suffice for Lemma~\ref{lem:subgradients}.

For the forward direction, assume that \scorename \ASS\ elicits $\Gamma$.  For each $r$, we may extend $\AS{r,\cdot}$ to all $\hat t\in \convhull(\T)$ by linearity as in the proof of Theorem~\ref{thm:main-char}, whence we may define $G(\hat t)\defeq \sup_{r\in \R} \AS{r,\hat t}$, which is finite for $\hat t\in \T$ as \ASS\ is $\Gamma$-regular.  We wish to show that the choice $\varphi : r \mapsto \AS[_\ell]{r,\cdot}$ suffices, where $\ASS_\ell$ denotes the linear part of $\ASS$, with $\D$ the range of $\varphi$ and $\{t_r\}$ arbitrary satisfying the theorem.  Given this construction, we need to check each of the following.

\paragraph{1. $G$ is convex with subgradients $\varphi(\Gamma(t)) \subseteq \partial G_t$}
Let $t$ and $r \in \Gamma(t)$ be given.  We show that $\varphi(r)$ satisfies the property of a subgradient at $t$, and thus $G$ is convex with appropriate subgradients.
\begin{align}
G(t) + \varphi(r)(t' - t)
&= \sup_{r' \in \R} \AS{r',t} + \AS[_\ell]{r,t' - t} \nonumber\\
&= \AS{r,t} + \AS[_\ell]{r,t' - t} = \AS{r,t'} \label{eq:prop-char-1}\\
&\leq \sup_{r' \in \R} \AS{r',t'} = G(t') \nonumber
\end{align}

\paragraph{2. $\ASS$ satisfies eq.~\eqref{eq:prop-char}}
This follows from~\eqref{eq:prop-char-1} with $t=t_r$, as $r \in \Gamma(t_r)$.

\paragraph{3. $\varphi$ is a bijection}
By definition, $\D$ is the range of $\varphi$, so we only need to check that it is injective.  Suppose for contradiction that $\varphi(r) = \varphi(r')$.  Then, by definition, $\AS[_\ell]{r,\cdot} = \AS[_\ell]{r',\cdot}$.  Since $\ASS$ elicits $\Gamma$, we have  $\AS{r,\cdot} = \AS{r',\cdot}$.  But then $r \in \Gamma(t) \iff r' \in \Gamma(t)$, contradicting $\Gamma$ being non-redundant.

\paragraph{4. $\Gamma(t) = \varphi^{-1}(\D\cap \partial G_t)$}
We already know that $\varphi(\Gamma(t)) \subseteq \partial G_t$, so since $\D$ is the range of $\varphi$ we have $\varphi(\Gamma(t)) \subseteq \D\cap\partial G_t$.  For the other direction, $d \in \D\cap\partial G_t$ is $\varphi(r)$ for some $r$.  Then by Lemma~\ref{lem:subgradients},
$\AS{r,t} = G(t_r) + \varphi(r)(t - t_r) = G(t),$ so $r \in \Gamma(t)$.
\end{proof}

As a corollary, we also obtain a better understanding of weak elicitation, which we will need in the following sections.

\begin{corollary}
  \label{cor:properties-weakly-elicits}
  Let non-redundant property $\Gamma:\T\toto \R$ and $\Gamma$-regular \scorename $\ASpropdef$ be given.  Then \ASS\ weakly elicits $\Gamma$ if and only if $\ASS$ satisfies~\eqref{eq:prop-char} with the weaker condition that $\Gamma(t) \subseteq \varphi^{-1}(\D\cap \partial G_t)$.
\end{corollary}
\begin{proof}
  We will use the equivalent condition $\varphi(\Gamma(t))\subseteq\D\cap\partial G_t$.  Also, recall the definition of $\Gamma^\ASS$ in~\eqref{eq:properties-3}, and note that for any \scorename $\ASS$, we have $\ASS$ weakly elicits $\Gamma$ if and only if $\Gamma(t) \subseteq \Gamma^\ASS(t)$ for all $t$; that is, if every report in $\Gamma(t)$ maximizes $\AS{\cdot,t}$ (but there could be others).

($\Rightarrow$)
Let $\ASS$ weakly elicit $\Gamma$, which by the above means $\Gamma \subseteq \Gamma^\ASS$.  Applying Theorem~\ref{thm:prop-char} for $\Gamma^\ASS$ we have some $G, \varphi, t_r$ such that we can write $\AS{r,t} = G(t_r) + \varphi(r)(t-t_r)$ for $\varphi(\Gamma^\ASS(t)) = \D\cap\partial G_t$ and $r \in \Gamma^\ASS(t_r)$.  This immediately gives $\varphi(\Gamma(t)) \subseteq \varphi(\Gamma^\ASS(t)) = \D\cap\partial G_t$.
It remains to show that $r\in\Gamma(t_r)$, which we will show holds without loss of generality: if $r\notin\Gamma(t_r)$, we at least have $\varphi(r) \in \partial G_{t_r}$ by the above, and we have some $\hat t_r$ such that $r\in\Gamma(\hat t_r)$ and $\varphi(r) \in \partial G_{\hat t_r}$.  Now by Lemma~\ref{lem:subgradients} we see that $G(t_r) - G(\hat t_r) = \varphi(r)(t_r - \hat t_r)$, and thus we may replace $t_r$ by $\hat t_r$ without changing $\ASS$.

($\Leftarrow$)
Let $\ASS$ be of the form~\eqref{eq:prop-char} for $r\in\Gamma(t_r)$ and $\varphi(\Gamma(t)) \subseteq \D\cap\partial G_t$ for all $t,r$.  As $\ASS$ elicits $\Gamma^\ASS$, we have $\varphi(\Gamma^\ASS(t)) = \D\cap\partial G_t$ and as $\varphi$ is a bijection, $\Gamma(t) \subseteq \Gamma^\ASS(t)$ for all $t$.
\end{proof}

Using Corollary~\ref{cor:properties-weakly-elicits}, we see that an \scorename $\ASS$ is truthful if and only if it weakly elicits $\Gamma : t \mapsto \{t\}$.  Hence, Theorem~\ref{thm:prop-char} and Corollary~\ref{cor:properties-weakly-elicits} are actually generalizations of Theorem~\ref{thm:main-char}.  
Of course, we also obtain the following corollary characterizing non-redundant properties.
\begin{corollary}\label{cor:props-elicitable-D}
  Non-redundant $\Gamma:\T\toto \R$ is elicitable if and only if exists there some convex $G : \convhull(\T)\to\extreals$ with $G(\T)\subseteq\reals$, some $\D\subseteq \partial G$, and some invertible $\varphi : \R \to \D$ such that $\Gamma(t) = \varphi^{-1}(\D\cap \partial G_t)$.
\end{corollary}

An important question which would give stronger characterizations is the following:
\begin{question}
\label{q:main}
  Given non-redundant elicitable $\Gamma$, what are all pairs $G,\D$ such that there exists
  some bijection $\varphi$ satisfying Theorem~\ref{thm:prop-char}?  Equivalently (up to redundancy), given a convex function $G$ with subgradient level sets $LS_G(d) = \{t : d \in \partial G_t \}$, what are all the convex functions $G'$ with $LS_{G'} \equiv LS_G$?
\end{question}

In Section~\ref{sec:props-finite} we will see that the answer to this question has a lot of structure in the case where $\R$ is finite.  In the general case, certainly performing a homothet of the subgradients of $G$ (i.e. scaling $G$ and adding a linear term), will preserve the elicitation structure.  However, surely more can be done---the property 
$\Gamma(t) = \{t\}$
can be elicited with both $G(t) = t^2/2$ and $G(t) = \abs{t} + t^2/2$, which is not a homothet transformation.

While we do not have a complete answer to Question~\ref{q:main} our characterization sheds new light on the structure of elicitable properties in two directions.  First, in the scoring rules literature, it is common to assume strong conditions on $\Gamma$ and $\R$, such as $\Gamma$ being a function rather than a multivalued map, and $\Gamma$ being linear~\cite{abernethy2012characterization} or real-valued~\cite{lambert2008eliciting} to achieve characterizations.  In contrast, Theorem~\ref{thm:prop-char} allows for an extremely general $\Gamma$ and $\R$ and shows us how to construct affine scores for such properties.  Second, we can identify features that all elicitable properties share, which provides a means to prove that specific properties are not elicitable.

\subsection{What Properties Are Not Elicitable?}

\label{sec:notelicitable}
In the remainder of this section, we examine three features that subgradient mappings of convex functions possess and thus that the level sets of elicitable properties must possess.

\subsubsection*{Convexity}

A well-known property of subgradient mappings is that their level sets are convex (for completeness, we provide a proof in Appendix~\ref{sec:cvx-primer}).

In light of our characterizations, this fact about convex functions immediately applies to elicitable properties.
\begin{proposition}
  \label{prop:properties-subgrad-level-sets}
  For convex functions $G$, the set $\partial G^{-1}(d) \defeq \{x \in \dom(G) : d\in\partial G_x\}$ is convex.
\end{proposition}
\begin{corollary}
  \label{cor:properties-level-sets}
  If $\Gamma:\T\toto\R$ is elicitable, then $\Gamma_r$ is convex for all $r$.
\end{corollary}
To see the corollary, just note that $\varphi(r) \in \partial G_t\cap \partial G_{t'}$ implies that $\varphi(r) \in \partial G_{\hat t}$ for all $\hat t = \alpha t + (1-\alpha) t'$.  Corollary~\ref{cor:properties-level-sets} was previously known for special cases~\cite{lambert2008eliciting,lambert2009eliciting}, where it was used to show variance, skewness, and kurtosis are not elicitable, and was also known in mechanism design (i.e. the set of types for which a given (allocation,\,payment) pair is optimal is convex).  

\subsubsection*{Cardinality}

Combining Theorem~\ref{thm:prop-char} with the fact that finite-dimensional convex functions are differentiable almost everywhere (cf.~\cite[Thm 7.26]{aliprantis2007infinite}) yields the following corollary, which shows that elicitable properties have unique values almost everywhere.

\begin{corollary}
  Let $\Gamma:\T\toto \R$ be an elicitable property with $\T\subseteq\V = \reals^n$.  If $\T$ is of positive measure in $\convhull(\T)$, and $\Gamma$ is non-redundant, then $|\Gamma(t)|=1$ almost everywhere.
\end{corollary}
In some cases, this statement holds in infinite-dimensional vector spaces as well, provided one uses an appropriate notion of ``almost everywhere'' (cf.~\cite[p. 195]{borwein2010convex} and~\cite[p. 274]{aliprantis2007infinite}).  One can use this fact to show that $\Gamma(p) = \{(a,b) : \int_a^b p(x)dx = 0.9\}$, the set of 90\% confidence intervals for a distribution $p$, is not an elicitable property.  This was previously only known for the case where $p$ has finite support~\cite{lambert2008eliciting}.

\subsubsection*{Topology}

Finally, we can use tools from topology to draw other conclusions.  Combining Theorem~\ref{thm:prop-char} with a closure property of convex functions~\cite[Thm 24.4]{rockafellar1997convex} yields the following.
\begin{corollary}
Let $\Gamma:\T\toto \R$ be an elicitable property with $\T\subseteq\V = \reals^n$ convex that can be elicited by a closed, convex $G$.  Then $\Gamma_r$ is closed for all $r$.  
\end{corollary}
Requiring $G$ to be closed is a generally mild technical assumption regarding the boundary of of $\T$, and is irrelevant for level sets in the relative interior.
While~\cite{lambert2009eliciting} showed this for finite report spaces $\R$, this more general statement shows, for example, that if $\T = \reals$ the property $\Gamma(t) = \mathrm{floor}(t) = \max \{ z \in \integers ~|~ z \leq t\}$ is not elicitable.  More generally, this often provides a tool for showing that we cannot get around issues of cardinality by finding a tie-breaking rule to make the value unique.

Closure appears is a more delicate attribute to work with in infinite dimensions, but intuitive violations of it can still be used to show that properties are not elicitable.
As an illustration, we provide a direct proof that a property is not elicitable.
We already saw that confidence intervals are not elicitable due to their cardinality, and therefore a natural way to break these ties, which is also of practical interest, would be to ask for the smallest such interval.
Can we elicit this smallest interval, or even just its length?
The following sketch shows we cannot.
For probability distribution $F$ represented by a CDF, let $\Gamma(F) = \inf \{b - a \mid F(b) - F(a) = 1/2\}$ be the length of the smallest interval of probability $1/2$.
Consider the family of distributions with densities $f_c$ defined as follows: $f_c(x) = 1/2 - c$ for $0 \leq x \leq 1$, $f_c(x) = c$ for $1 < x \leq 2$, and $f_c(x) = 1/4$ for $4\leq x \leq 6$, with corresponding CDFs $F_c$.
Note that for $0 < c < 1/2$, $\Gamma(F_c) = \{2\}$ but $\Gamma(F_0) = \{1\}$.
Suppose we could elicit $\Gamma$ with scoring rule $\SRS$.
Let $X(F)= \SR{2,F} - \SR{1,F}$.
By elicitability $X(F_c) > 0$ for $0 < c < 1/2$ but $X(F_0) < 0$, which violates the continuity of $X$ given by linearity of expectation.
This argument can be extended to arbitrary properties as follows.

\begin{corollary}
  \label{cor:affine-line-continuity}
  Let $\Gamma:\T\toto \R$ be an elicitable property for $\T$ convex, and let $t,t'\in\T$ and $t_\alpha = \alpha t + (1-\alpha)t'$ for $\alpha \in [0,1]$.  If $\Gamma(t_\alpha) = \{r\}$ for $\alpha \in [0,1)$, then $r\in\Gamma(t)$.  
\end{corollary}
\begin{proof}
  Suppose \ASS\ elicits $\Gamma$, and assume for a contradiction that $r\in\Gamma(t_\alpha)$ for $\alpha \in [0,1)$, but $r\notin\Gamma(t)$.  Let $r'\in\Gamma(t)$, noting $r'\neq r$, and define $X(\alpha) \defeq \AS{r,t_\alpha}-\AS{r',t_\alpha}$.  By affineness of \ASS, we can write $X(\alpha) = \alpha(\AS{r,t}-\AS{r',t}) + (1-\alpha)(\AS{r,t'}-\AS{r',t'})$, which is continuous with respect to $\alpha$.  But now $X(\alpha)>0$ for $\alpha<1$ as $r'\notin\{r\}=\Gamma(t_\alpha)$, while $X(1) < 0$ as $r\notin\Gamma(t)\supseteq\{r'\}$, contradicting continuity of $X$.
\end{proof}

Using Corollary~\ref{cor:affine-line-continuity}, we can extend the argument above to show that neither the smallest confidence interval containing mass $\alpha$, defined as $\Gamma_\alpha(F) = \argmin_{(a,b)} \{ b-a ~|~ F(b)-F(a)=\alpha \}$, nor its length, is elicitable with respect to families of distributions which include all piecewise linear CDFs with a finite number of pieces.

\section{Duality in Property Elicitation}
\label{chap:propertiesduality-elicitation}

In the previous section, we saw from Theorem~\ref{thm:prop-char} that in a strong sense an elicitable property $\Gamma$ is like a subgradient mapping of a convex function.
In order to explore notions of duality, we first remove the word ``like'' from the previous sentence---we look at properties which \emph{are} subgradient mappings.
This exercise has two main benefits.
First, it gives us a concrete tool to reason about properties, by working directly with a convex function rather than through some map $\varphi$.
Second, it gives a new framework to discuss duality in elicitation, as has been observed between scoring rules and prediction markets~\cite{hanson2003combinatorial,chen2010new,abernethy2012characterization,chen2013cost,abernethy2013efficient}.

\subsection{Subgradient Elicitation}
\label{chap:propertiessubgradient-elicitation}

Now that we have formalized the relationship between the report space and subgradients of convex functions, we can see what the ``canonical'' properties look like: those which are (subsets of) subgradient mappings of a convex function.  For these properties, we can talk about \emph{subgradient elicitation}, which roughly speaking means removing the intermediary $\varphi$ between $\R$ and $\partial G$.  In fact, for such properties, it makes sense to talk about a convex function itself eliciting~$\Gamma$.

  \begin{definition}
    A property $\Gamma:\T\toto\D$, where $\T\subseteq\V$ and $\D\subseteq\V^*\defeq\linear(\V\to\extreals)$, is \emph{subgradient-elicitable} if there exists $G:\conv(\T)\to\extreals$ convex with $G(\T)\subseteq \reals$ such that $\Gamma(t)\subseteq \partial G_t$.  In this case we say $G$ \emph{subgradient-elicits}, or just \emph{elicits}, $\Gamma$.
  \end{definition}

  In other words, $G$ elicits $\Gamma:\T\toto \D$ if the $\varphi$ in Theorem~\ref{thm:prop-char} and Corollary~\ref{cor:props-elicitable-D} is the identity.
  It is easy to construct an \scorename eliciting such a property.

  \begin{proposition}
    Subgradient-elicitable properties are elicitable.
  \end{proposition}
  \begin{proof}
    Let $\Gamma:\T\toto \R$ and $G:\conv(\T)\to\extreals$ convex with $G(\T)\subseteq \reals$ be given such that $\Gamma(t)\subseteq \partial G_t$.  Then taking $\D = \R$ and $\varphi = \id_\D$, we have by Corollary~\ref{cor:props-elicitable-D} that $\Gamma$ is elicitable.
Note that despite  $\Gamma(t)\subseteq \partial G_t$ we have $\Gamma(t) = \argsup_{r\in \R} \AS{r,t}$ because the range of $\Gamma$ is $\R$.
  \end{proof}
  Note that this subgradient elicitability in no way necessary for elicitability, since the report space can be arbitrary.
  For example, one can take $\Gamma(t) \defeq -\partial G_t$ for some $G$, which in general will not be subgradient-elicitable, but still elicitable with $\varphi(r) = -r$ and $G$.

We can also clarify what we mean when we say subgradient elicitation is canonical: every elicitable property gives rise to a subgradient-elicitable property.
\begin{proposition}
Let $\Gamma$ be an elicitable property, elicited by $\AS{r,t} = G(t_r) + \varphi(r)(t - t_r)$.  Then $\Gamma^\varphi(t) = \varphi(\Gamma(t))$ is subgradient-elicitable.
\end{proposition}
\begin{proof}
Simply keep $G$ and take $\id_\D$ as the new $\varphi$.
\end{proof}
In other words, properties are subsets of subderivative mappings, up to some bijection, or \emph{link function}, taking them to some other report space $\R'$; see also the discussion following Theorem~\ref{thm:dual-report-char}.

As a final remark, we note a few observations about subgradient elicitation.  One first notices that the $G$ eliciting some $\Gamma$ is not unique, as $G' \defeq G + c$ will also elicit $\Gamma$ for any constant $c$.  But these are the \emph{only} convex functions subgradient-eliciting $\Gamma$.  Moreover, recovering such a $G$ from $\Gamma$ is easy: simply integrate (a selection of) $\Gamma$ to obtain $G$.  Testing whether $\Gamma$ is subgradient-elicitable is less straight-forward, but there are a variety of monotonicity conditions addressing this issue as well (cf Appendix~\ref{sec:convex}).

\subsection{Report Duality}
\label{chap:propertiesreport-duality}

We are now ready to introduce the first of two notions of duality in property elicitation, where we change the report from the type to the dual type, but leave the type itself unchanged.
For now, we will take our dual vector space to be all linear functions from $\V$ to $\reals$ (not $\extreals$ as above).\footnote{When the dual space can take on infinite values, the conjugate is not always well-defined, as values of the form $\infty-\infty$ are encountered.}  We begin with the fundamental operation of convex duality, the convex conjugate.

\begin{definition}
  \label{def:properties-conjugate}
  Let $\V^* \defeq \linear(\V\to\reals)$.  The \emph{convex conjugate} of $G:\V\to\extreals$, denoted $G^*:\V^*\to\extreals$, is given by
  \begin{equation}
    \label{eq:conv-conj}
    G^*(d) = \sup_{v\in \V}\; d(v) - G(v).
  \end{equation}
\end{definition}

The power of the conjugate is apparent from Fact~\ref{lem:properties-dual-optimal} in Appendix~\ref{sec:cvx-primer}, which says that $v$ is an optimum of~\eqref{eq:conv-conj}, meaning $G^*(d) = d(v) - G(v)$, if and only if $d$ is a subgradient at $v$.  In this way, the convex conjugate in some sense encodes the subgradients of $G$.  Fact~\ref{lem:properties-dual-optimal} lets us further simplify Theorem~\ref{thm:prop-char}, as follows.  Note however that we are making an additional assumption, that $G > -\infty$.

\begin{theorem}
  \label{thm:dual-report-char}
  Let non-redundant property $\Gamma:\T\toto \R$ and $\Gamma$-regular \scorename $\ASpropdef$ be given with score set $\A\subseteq\affine(\T\to\reals)$.  Then \ASS\ elicits $\Gamma$ if and only if there exists some convex $G : \convhull(\T)\to\reals$, and bijective $\varphi : \R \to \D$ with $\D\subseteq \partial G$ satisfying $\varphi(\Gamma(t))\subseteq\partial G_t$, such that for all $r\in \R$ and $t\in\T$,
  \begin{equation}
    \label{eq:dual-prop-char}
    \AS{r,t} \; = \; \varphi(r)(t) - G^*(\varphi(r)).
  \end{equation}
\end{theorem}

Theorem~\ref{thm:dual-report-char} has natural interpretations for both mechanisms and scoring rules.  For mechanisms, it captures a version of the {\em taxation principle}, that a mechanism can be viewed as a menu of possible allocations $\varphi(r)$ with corresponding prices $G^*(\varphi(r))$.  It is worth noting however that this posted-price mechanism is not always identical to the original mechanism.  Specifically, while the equilibrium payoffs for the posted-price mechanism $\AS{r,t}$ are the same as those of the direct revelation mechanism $\AS{t',t}$, the \emph{off-equilibrium} payoffs need not be equivalent, as the posted-price mechanism may allow reports $\varphi(r)\in\partial G_t$ which are not $dG_{t'}$ for any $t'$.  In other words, because the primal-report (i.e., direct) mechanism must choose a single subgradient $dG_t$ for every point, if $\{dG_t\}_\T \subsetneq \partial G = \D$, the dual-report mechanism can be strictly more expressive.

For scoring rules, Theorem~\ref{thm:dual-report-char} captures the relationship between a scoring rule and a prediction market.
The prediction market model we refer to is the standard cost-function framework~\cite{abernethy2013efficient}, which we briefly describe.
A centralized market maker chooses a convex cost function $C$, and traders who wish to buy a bundle of securities $q\in\reals^\O$ (where the trader will receive \$$q_\o$ upon outcome $\o$) pays $C(q^0+q) - C(q^0)$, where $q^0$ is the vector of total purchases made so far in the market.  Abstracting away $q^0$, we set $G^*(q) = C(q^0+q) - C(q^0)$, yielding affine score $\AS{q',p} = \inprod{p, q'} - G^*(q')$.
Hence, scoring rules have the primal report space, and prediction markets the dual, where an optimal share bundle is essentially a subgradient of the scoring rule at the trader's belief.
As with posted-price mechanisms, the prediction market can be strictly more expressive.


\subsection{Type Duality and the Duality Quadrangle}
\label{sec:propertiestype-duality-duality}

We have seen the first notion of duality, where the report switches from the type space to the dual space.
In Appendix~\ref{sec:append-duality} we introduce the second notion, that of dual \emph{properties}, where we swap the roles of types and reports.
Taken together, these two duality notions give four combinations of dual report and type, which form a ``duality quadrangle'', depicted in Table~\ref{tab:quadrangle-simple}.
We focus here on the particular instantiation of type space $\Delta(\O)$ and dual type space $(\O\to\reals)$, based on the classic duality between integrable functions and probability measures; Appendix~\ref{sec:append-duality} generalizes to any dual pair of vector spaces.

\definecolor{dist}{rgb}{0,.6,0}
\definecolor{util}{rgb}{0,0,.6}
\newcommand{\myp}{\textcolor{dist}{p}}
\newcommand{\myq}{\textcolor{util}{q}}
\newcommand{\mypp}{\textcolor{dist}{p'}}
\newcommand{\myqq}{\textcolor{util}{q'}}

\begin{table}[ht]
  \centering
  \def\dualitycellwidth{160pt}
  \begin{tabular}{c|c|@{}c@{}|@{}c@{}|}
    \multicolumn{1}{r}{} & \multicolumn{1}{r}{} & \multicolumn{2}{c}{Private Type}
    \\[4pt]
    \cline{3-4}
    \multicolumn{1}{r}{} & \multicolumn{1}{r|}{} &
    \textcolor{dist}{Distribution ($p$)} \bigstrut & \textcolor{util}{Utility ($q$)}\\[2pt]
    \cline{2-4}
    \multirow{2}{*}{\parbox[c][100pt]{10pt}{\rotatebox{90}{Report}}} &
    \parbox[c]{10pt}{\rotatebox{90}{\textcolor{dist}{Distrib. ($\mypp$)}}} &
    \parbox[c][70pt]{\dualitycellwidth}{\centering $\AS{\mypp,\myp}$ $=$ $G(\mypp) + \inprod{\myp-\mypp,dG_{\mypp}}$\\[5pt]Scoring rule} &
    \parbox[c][70pt]{\dualitycellwidth}{\centering $\AS[^*]{\mypp,\myq}$ $=$ $\inprod{\mypp,\myq} - G(\mypp)$\\[5pt]Posted-price mechanism}\\ 
    \cline{2-4}
    &
    \parbox[c]{10pt}{\rotatebox{90}{\textcolor{util}{Utility ($\myqq$)}}} &
    \parbox[c][70pt]{\dualitycellwidth}{\centering $\AS{\myqq,\myp}$ $=$ $\inprod{\myp,\myqq} - G^*(\myqq)$\\[5pt]Prediction market} &
    \parbox[c][70pt]{\dualitycellwidth}{\centering $\AS[^*]{\myqq,\myq}$ $=$ $G^*(\myq) + \inprod{dG^*_{\myqq},\myq-\myqq}$\\[5pt]Randomized mechanism}\\
    \cline{2-4}\noalign{\vskip 5pt}
    \multicolumn{1}{c}{}& \multicolumn{1}{c}{}&
    \multicolumn{1}{c}{$\sup\: \AS{\cdot,\myp} = G(\myp)$} &
    \multicolumn{1}{c}{$\sup\: \AS[^*]{\cdot,\myq} = G^*(\myq)$}
    \vspace{5pt}
  \end{tabular}
  \caption{The duality quadrangle for the duality between distributions and functionals.}
  \label{tab:quadrangle-simple}
\end{table}

As we discussed in Section~\ref{chap:propertiesreport-duality}, the columns of Table~\ref{tab:quadrangle-simple} are well-understood already. The first gives prediction market duality, the well-known fact that market scoring rules are dual to prediction markets, in the sense of report duality.
The second gives the taxation principle, which says that without loss of generality one could think of a mechanism as assigning prices to probability distributions over outcomes $\o$.
The rows of this table, however, are new, as we swap the nature of the private information itself.
In the scoring rule or prediction market setting, an agent has a private distribution (their belief) and the principal gives the agent a utility vector (the score or the bundle of securities), which assigns the agent a real-valued payoff for each possible state of the world.
Dually, in a mechanism, the agent possesses a private type encoding their utility for each state of the world, and the principal assigns a distribution over these states.

Thus, in essence, scoring rules are dual mechanisms.
In the remainder of this section, we elaborate on this particular connection, in the finite-outcome setting for clarity, giving a concrete construction to convert scoring rules to mechanisms and vice versa.
Note that unlike previous constructions (e.g.,~\cite{fiat2013approaching}) we can establish a direct bijection between the two objects, and do not require any normalization.

Fix a finite outcome space $\O$.
Let $\mathcal{S}$ be the set of all scoring rules $\SRS:\P\times\O\to\extreals$ for some convex set of distributions $\P\subseteq\Delta(\O)$, and $\mathcal{M}$ the set of all randomized mechanisms $(f:\T\to\Delta(\O),\pi:\T\to\reals)$ for some convex type space $\T\subseteq\reals^\O$.
(We use $\pi$ for prices to avoid confusion with probabilities.)
In what follows, we will write expectations as dot products, such as $t\cdot f(t') = \E_{\o\sim f(t')}[t(\o)]$.

\begin{construction}
  \label{construction}
  Given a scoring rule $\SRS:\P\times\O\to\extreals$, let $G_\SRS:\P\to\extreals$ be defined by $G_\SRS(p) = \sup_{p'\in\P} \SR{p',p}$.
  Similarly, given a randomized mechanism $\MDS = (f,\pi)$ on type space $\T$,
  let $G_\MDS:\T\to\extreals$ be given by $G_\MDS(t) = \sup_{t'\in\T} U(t',t) = \sum_{t'\in\T} t\cdot f(t') - \pi(t')$.
  Define the multivalued maps $\Phi : \mathcal{S} \toto \mathcal{M}$ and $\Psi : \mathcal{M} \toto \mathcal{S}$ as follows,
  \begin{equation}
    \label{eq:construction}
    \begin{aligned}
      \Phi(\SRS) &= \{ \MDS=(f,\pi) \mid \forall t\in\T\; f(t) \in \partial G_\SRS^*(t),\; \pi(t) = t \cdot f(t) - G_\SRS^*(t)\}
      \\
      \Psi(\MDS) &= \{ \SR{p,\o} = G_\MDS^*(p) + d_p(\o) - d_p \cdot p \mid \forall p\in\P\; d \in \partial G_\MDS^*(p)\}~,
    \end{aligned}
  \end{equation}
  where for any function $F:D\to\extreals$ on domain $D\subseteq \reals^d$ we write $F^*:\reals^d\to\extreals$ for the convex conjugate, $F^*(x) = \sup_{d\in D} d\cdot x - F(d)$.
\end{construction}

Our construction takes a scoring rule and produces a set of mechanisms, and vice versa.
We now show that these are inverse operations, under mild regularity assumptions.
The set-valued nature of our construction merely follows from the choices in selecting subgradients in Theorem~\ref{thm:main-char}; as we show in part (3) below, when there is only one choice for the subgradients, the maps $\Phi,\Psi$ become single-valued.

\begin{theorem}
  \label{thm:construction} 
  We have the following:
  \begin{enumerate}
  \item For any regular proper scoring rule $\SRS$, any $\MDS\in\Phi(\SRS)$ is a truthful mechanism;
    conversely, for any truthful mechanism $\MDS$, any $\SRS\in\Psi(\MDS)$ is a regular proper scoring rule.
  \item  Let $\mathcal{C}$ be the set of closed\footnote{A convex function is closed if its epigraph is a closed set, or equivalently, if it is lower semi-continuous.} convex functions.  Restricting to the sets $\{\SRS\in\mathcal{S}:G_\SRS\in\mathcal{C}\}$ and $\{\MDS\in\mathcal{M}:G_\MDS\in\mathcal{C}\}$, the maps $\Phi$ and $\Psi$ are inverses as multivalued maps.
  \item Let $\mathcal{D}$ the set of closed strictly convex functions $f$ which are differentiable in the interior $I$ of their domain, and satisfy $\lim_{k\to\infty} \|\nabla f(x_k)\| = \infty$ for any sequence $x_1,x_2,\ldots \in I$ converging to a boundary point of $I$. Restricting to the sets $\{\SRS\in\mathcal{S}:G_\SRS\in\mathcal{D}\}$ and $\{\MDS\in\mathcal{M}:G_\MDS\in\mathcal{D}\}$, $\Phi$ and $\Psi$ are single-valued (and inverses).
    Furthermore, every scoring rule and mechanism in these sets is strictly proper/ truthful.
  \end{enumerate}
\end{theorem}
\begin{proof}
  For part (1), note that by Corollaries~\ref{thm:scoring-rule-char} and~\ref{cor:archerkleinberg}, the forms are proper/truthful if we can show that the convex functions are finite on the correct domain, and subdifferentiable in the case of $G_\SRS^*$ (recall that we allow vertical subgradients for scoring rules).
  First, let $\SRS$ be given.
  Let $\T = \partial G_\SRS(\P) \cap \reals^\O$ be the set of finite subgradients, which is a convex set.
  From \cite[Corollary 13.3.1]{rockafellar1997convex} we have that $G_\SRS^*$ is finite on all of $\reals^\O$ as the domain of $G_\SRS$ is bounded.
  (The Corollary requires $G_\SRS$ to be closed, but we also have $G^*_\SRS = (\mathrm{cl}\, G_\SRS)^*$, and the domain of $\mathrm{cl}\, G_\SRS$ is still bounded, where $\mathrm{cl}$ denotes the closure of a function.)
  Subdifferentiability of $G_\SRS^*$ is then implied by \cite[Theorem 23.4]{rockafellar1997convex}.
  Conversely, given some mechanism $\MDS$, we can show that that $G_\MDS^*$ is finite on $\P = \partial G_\MDS(\T) \cap \Delta(\O)$.
  Applying \cite[Theorem 23.5]{rockafellar1997convex}, we have $p \in \partial G_\MDS(t) \iff G_\MDS(t) + G_\MDS^*(p) = p \cdot t$, and as $p\cdot t$ and $G_\MDS(t)$ are both finite, we conclude that $G_\MDS^*(p)$ is finite.
  Thus, by the Corollaries, we are done.

  For (2), note that $G^*_\SRS$ and $G^*_\MDS$ are closed and convex, and we have both $(G_\SRS^*)^* = G_\SRS$ and $(G_\MDS^*)^* = G_\MDS$~\cite[Theorem 12.2]{rockafellar1997convex}.
  The inverse relationship then follows from the fact that the multivalued maps~\eqref{eq:construction} depend only on $G_\SRS$ and $G_\MDS$, and every scoring rule / mechanism in the range shares the same expected score / consumer surplus function.

  For (3), we simply appeal to~\cite[Theorem 26.3 \& Corollary 26.3.1]{rockafellar1997convex}.
  Strictness follows from strict convexity.
\end{proof}

To illustrate Theorem~\ref{thm:construction}, consider the log scoring rule $\SR{p,\o} = \log p(\o)$, which is strictly proper.
The expected score function is negative Shannon entropy, $G_\SRS(p) = \sum_{\o\in\O} p(\o) \log p(\o)$, which satisfies the conditions of part (3) of Theorem~\ref{thm:construction}.
Applying $\Phi$, we therefore obtain a unique mechanism, which chooses the allocation probabilites according to the familiar multiplicative weights formula:
\begin{equation}
  \label{eq:1}
  f(t) = \left(\tfrac{e^{t(\o)}}{\sum_{\o'\in\O}e^{t(\o')}}\right)_{\o\in\O} \in \Delta(\O),\;\;\;
  \pi(t) = \log\sum_{\o\in\O} e^{t(\o)} - \tfrac{\sum_{\o\in\O} t(\o) e^{t(\o)}}{\sum_{\o\in\O}e^{t(\o)}}~.
\end{equation}
The prices $\pi$ bear resemblance to the Log Market Scoring Rule (LMSR)~\cite{hanson2003combinatorial}, but as remarked above, note that the private information here is a fixed valuation function $t$, not a belief in the form of a probability distribution.

The bijection in Theorem~\ref{thm:construction} is not an arbitrary one; a mechanism and scoring rule which are related by $\Phi$ or $\Psi$ satisfy certain identities.
For example, for $\SRS$ and $\MDS=(f,\pi)$ satisfying part (2) of the theorem, and $\MDS\in\Phi(\SRS)$, a standard result of convex analysis states that $G_\SRS(p) + G_\MDS(t) = p\cdot t$ whenever either $t = d$ (from eq.~\eqref{eq:construction}) or $p = f(t)$ hold.
To provide further intuition for this relationship with a somewhat whimsical example, suppose a gambler in a casino examines the rules of a dice-based game of chance and forms belief $p$ about the probabilities of possible outcomes, assuming the dice are fair.
The gambler is then offered a proper scoring rule $\SRS$ to predict the outcome of the game, and reports truthfully.
Before the game is played, however, the casino informs the gambler that the dice used need not be fair, and offers the gambler the opportunity to select from among different choices of dice using a truthful mechanism $\MDS$ where the gambler's private information is the valuation function $t$ given by $t(\o) = \SR{p,\o}$;\footnote{Strictly speaking, to directly apply the convex analysis result referenced above, we would need to take $t(\o)$ to be the linear part of  $\SR{p,\o}$; in our setting, however, the allocation and prices of the mechanism are invariant to adding a constant to the valuation of every outcome.  (To see this, note that the subgradients of $G_\MDS = G_\SRS^*$ are probability distributions, so the directional derivative is 1 in the  direction of $\ones$, the all-ones vector; integrating gives the result.  See also~\cite{abernethy2013efficient}.)\label{foot:ones-invariance}}
 note that $t(\o)$ truly is the gambler's valuation of outcome $\o$.
If $\MDS \in \Phi(\SRS)$,
and the gambler again reports truthfully,
then the dice chosen by the mechanism will be fair.
And what will be the gambler's profit in expectation from both the scoring rule and mechanism?
Zero.
This follows from the above identity, and the observation that $t \cdot f(t) - \pi(t) = \SR{p,p} - (p\cdot t - G_\MDS(t)) = G_\SRS(p) + G_\MDS(t) - p\cdot t$.\footref{foot:ones-invariance}
The power of our construction is that these relationships hold regardless of the initial scoring rule $\SRS$.

\section{Finite-Valued Properties}
\label{sec:props-finite}

We now examine the special case where $\R$ is a finite set of reports, using the additional structure to provide stronger characterizations.
In the scoring rules literature, Lambert and Shoham~\citeyear{lambert2009eliciting} view this as eliciting answers to multiple-choice questions.  There are also applications to mechanism design, discussed in Section~\ref{sec:saksyu}.
Assume throughout that $\R$ is finite and that $\T$ is a convex subset of a Hilbert space $\V$ (e.g., $\V = \reals^d$) with inner product written $\inprod{t,t'}$ and norm $\|t\|^2 = \inprod{t,t}$.  In this setting, we will use the concept of a power diagram from computational geometry.

\begin{definition}
  Given a set of points $P = \{p_i\}_{i=1}^m \subset \V$, called \emph{sites}, and weights $w\in\reals^m$, a \emph{power diagram} $D(P,w)$ is a collection of cells $\mathrm{cell}(p_i) \subseteq \T$ defined by
  \begin{equation}
    \label{eq:power-cell}
    \mathrm{cell}_{P,w}(p_i) = \left\{ t\in\T \,\left|\, i \in \argmin_j\left\{ \|p_j - t\|^2 - w_j \right\}\right.\right\}.
  \end{equation}
\end{definition}

The following result is a straightforward generalization of Theorem 4.1 of Lambert and Shoham~\citeyear{lambert2009eliciting}, and is essentially a restatement of results due to Aurenhammer~\citeyear{aurenhammer1987criterion,aurenhammer1987power}. 
\begin{theorem}
  \label{thm:power-diag}
  A property $\Gamma : \T \toto \R$ for finite $\R$ is elicitable if and only if the level sets $\{\Gamma_r\}_{r\in \R}$ form a power diagram $D(P,w)$.  
\end{theorem}
\begin{proof}
  Let us examine the condition that $t$ is an element of $\mathrm{cell}_{P,w}(p_i)$ for some power diagram $D(P,w)$:
  \begin{align}
    t \in \mathrm{cell}_{P,w}(p_i)
    & \iff i \in \argmin_j\left\{ \|p_j - t\|^2 - w_j \right\} \nonumber \\
    & \iff i \in \argmin_j\left\{ \|p_j\|^2 - 2\inprod{p_j,t} - w_j \right\}.
    \label{eq:properties-5}
  \end{align}
  Note that eq.~\eqref{eq:properties-5} is affine in $t$.  Now given some $D = D(P,w)$ with index set $\R$, we simply let $\AS{r,t} = 2\inprod{p_r,t} + w_r - \|p_r\|^2$.  By~\eqref{eq:properties-5} we immediately have $r \in \argsup_{r'} \AS{r',t} \iff t \in \mathrm{cell}_{P,w}(p_r)$, as desired.

  Conversely, let an \scorename \ASS\ eliciting $\Gamma$ be given.
  By the Fréchet--Riesz representation theorem, we may write $\AS{r,t} = \inprod{x_r,t} + c_r$ for $x_r \in \V$ and $c_r \in \reals$.
  Letting $p_r = x_r/2$ and $w_r = \|p_r\|^2 + c_r$, we see by~\eqref{eq:properties-5} again that $\Gamma_r = \mathrm{cell}(p_r)$ of the diagram $D(\{p_r\},w)$.
  Hence, $\Gamma$ is a power diagram.
\end{proof}

We have now seen what kinds of finite-valued properties are elicitable, but how can we elicit them?  More precisely, as the proof above gives sufficient conditions, what are all ways of eliciting a given power-diagram?  In general, it is difficult to provide a ``closed form'' answer to this question, so we restrict to the \emph{simple} case, where essentially the cells of a power diagram are as constrained as possible.

\begin{definition}[\cite{aurenhammer1987recognising}]
  \label{def:properties-simple}
  A \emph{$j$-polyhedron} is the intersection of dimension $j$ of a finite number of closed halfspaces of $\V = \reals^d$, where $0 \leq j \leq d$.  A \emph{tiling} $C$ in $\V$ is a covering of $\V$ by finitely many $j$-polyhedra, called \emph{$j$-faces} of $C$, whose (relative) interiors are disjoint.  If furthermore their non-empty intersections are faces of $C$ then $C$ is a \emph{cell complex}.   A cell complex $C$ is called \emph{simple} if each of its $j$-faces is in the closure of exactly $(d - j + 1)$ $d$-faces (cells).
\end{definition}

\begin{theorem}
  \label{thm:props-finite-score-char}
  Let $\V = \reals^d$ and let finite-valued, elicitable, simple property $\Gamma : \T \toto \R$ be given.
  Then there exist points $\{p_r\}_\R \subseteq \V$ such that the following holds: for any \scorename $\ASpropdef$ eliciting $\Gamma$, there exist $\alpha > 0$, $p_0 \in \V$, and $w \in \reals^\R$ such that
    \begin{equation}
      \label{eq:finite-char}
      \AS{r,t}
      = 2\inprod{\alpha p_r + p_0,t} - \|\alpha p_r + p_0\|^2 + w_r~,
    \end{equation}
    and conversely, for all such $\alpha$ and $p_0$ there exists $w \in \reals^\R$ making \ASS\ in eq.~\eqref{eq:finite-char} elicit $\Gamma$.
\end{theorem}
\begin{proof}
  A result of Aurenhammer for simple cell complexes, given in Lemma 1 of~\cite{aurenhammer1987power} and the proof of Lemma 4 of~\cite{aurenhammer1987criterion}, states the following: given sites $P$ and $P'$ and weights $w$, there exist weights $w'$ such that $D(P',w') = D(P,w)$ if and only if $P'$ is a homothet (translated and positively scaled copy) of $P$.  We simply apply this fact to the proof of Theorem~\ref{thm:power-diag}.
\end{proof}

\begin{figure}[htb]
  \centering
  \includegraphics[width=0.6\textwidth]{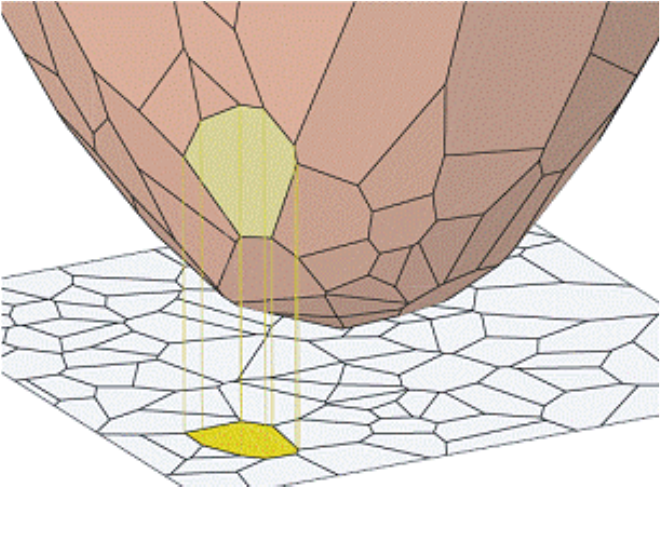}\\[-30pt]
  {\tiny \hspace{0.4\textwidth} Image credit: Camille Wormser}
  \caption{A consumer surplus function $G$ and its corresponding partition of the type space, $\Gamma$.  The proof of Theorem~\ref{thm:power-diag} leverages the fundamental relationship between projections of convex functions and power diagrams.}
  \label{fig:property-projection}
\end{figure}

See Appendix~\ref{sec:breg-vor} for a discussion about Bregman Voronoi digrams and the role of $\|\cdot\|^2$ in Theorem~\ref{thm:power-diag}.

\paragraph{Detecting elicitable finite properties}
As a practical matter, it is natural to ask if we can efficiently determine whether a given finite-valued property $\Gamma$ is elicitable.
By Theorem~\ref{thm:power-diag}, we need only test whether the cells $C = \{\Gamma_r\}_{r\in \R}$ form a power diagram.
For the \emph{simple} case, Aurenhammer gives the ``Orthogonal Dual'' algorithm for this task; see \sect~2.2 of~\cite{aurenhammer1987recognising} and comments thereafter.
The orthogonal dual algorithm assumes that the cells are stored in an \emph{incidence lattice}, with nodes for each face of $C$, and edges when faces are incident (a $j$-dimensional face which contains a $(j-1)$-dimensional face).
The runtime of the algorithm is $O(m)$, where $m$ is the number of facets (faces of dimension $d-1$).
More generally, Borgwardt and Frongillo~\citeyear{borgwardt2019power} present a weakly polynomial-time algorithm to detect power diagrams in the general case, via a simple linear program.

\subsection{Finite Properties in Mechanism Design}
\label{sec:saksyu}

Mechanisms with a finite set of allocations are common.  Carroll~\citeyear{carroll2012when} examines them and observes they give rise to polyhedral typespaces.  Theorem~\ref{thm:power-diag} strengthens this characterization to power diagrams, which rules out polyhedral examples such as the one shown in Figure~\ref{fig:polyhedral}.  In particular, the example of the left of the figure fails to be a power diagram because all power diagrams in $\reals^d$ are cell complexes~\cite{aurenhammer1987criterion}, while it is merely a tiling as the intersection of the $a_1$ and $a_2$ cells is not a face of the $a_1$ cell.

Suppose we are in a such a mechanism design setting with a finite set of allocations $\X$ and we have picked an allocation rule $a$.
Under what circumstances is $a$ implementable, i.e., when is there a payment rule that makes the resulting mechanism truthful?
For convex type sets, Saks and Yu~\citeyear{saks2005weak} showed that the following condition is necessary and sufficient.

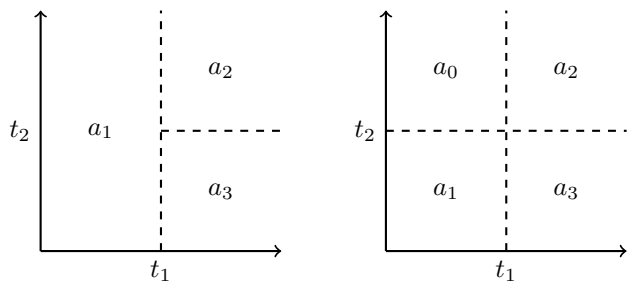
\begin{figure}
  \centering
  \begin{tabular}{ccc}
    \begin{tikzpicture}[thick,scale=0.8]
      \path[->] (0,0) edge node[left]{$t_2$} (0,4) edge node[below]{$t_1$} (4,0);
      \draw[dashed] (2,0) -- (2,4) (2,2) -- (4,2);
      \node[anchor=center] at (1,2) {$a_1$};
      \node[anchor=center] at (3,1) {$a_3$};
      \node[anchor=center] at (3,3) {$a_2$};
    \end{tikzpicture}
    & \quad &
    \begin{tikzpicture}[thick,scale=0.8]
      \path[->] (0,0) edge node[left]{$t_2$} (0,4) edge node[below]{$t_1$} (4,0);
      \draw[dashed] (2,0) -- (2,4) (0,2) -- (4,2);
      \node[anchor=center] at (1,1) {$a_1$};
      \node[anchor=center] at (1,3) {$a_0$};
      \node[anchor=center] at (3,1) {$a_3$};
      \node[anchor=center] at (3,3) {$a_2$};
    \end{tikzpicture}
\end{tabular}
  \caption{An allocation rule which cannot be implemented, for any distinct choices of $a_i$ (left), and a rule which could be implemented for appropriate choices of $a_i$ (right).}
  \label{fig:polyhedral}
\end{figure}

\begin{definition}
Allocation rule $a$ satisfies {\em weak monotonicity} (WMON) if $a(t) \cdot (t' - t) \leq a(t') \cdot (t'-t)$ for all $t,t' \in \T$.
\end{definition}
From Theorem~\ref{thm:main-char}, we know that $a$ being implementable means that there exists a $G$ such that $a$ is a selection of its subgradients. 
But this is equivalent to saying that the property $\Gamma(t) = \X \cap dG_t$ is subgradient-elicitable.  This gives us a new proof of this theorem by showing that WMON characterizes power diagrams.  In particular, we can leverage the following characterization of power diagrams due to Aurenhammer~\citeyear{aurenhammer1987criterion}.  This result assumes that $a$ is defined on all of $\V = \reals^d$, as until recently power diagrams have not been studied on restricted domains.
However, recent results for restricted domains imply that Aurenhammer's argument generalizes in a straightforward way~\cite{borgwardt2019power}.
For completeness, we provide a sketch of the proof for the restricted version.

\begin{definition}
  \label{def:properties-simple-restricted}
A \emph{tiling} $C$ of a convex set  $\T \subseteq \V$ is a covering of $\T$ by finitely many $j$-polyhedra, called \emph{$j$-faces} of $C$, whose (relative) interiors when restricted to $\T$ are disjoint.
\end{definition}

\begin{theorem}[\cite{aurenhammer1987criterion}]
  \label{thm:orthogonal-dual}
  Let $C$ be a tiling of a convex set $\T \subseteq \V = \reals^d$.
  
  Then $C$ is the restriction to $\T$ of a power diagram defined on all of $\V$, if and only if there exists a point-set $\{p_1,\ldots,p_n\}$ satisfying,
  \begin{enumerate}
  \item Orthogonality: For $Z_i \neq Z_j$, the line $L$ that contains $p_i$ and $p_j$ (and is directed from $p_i$ to $p_j$) is orthogonal to each face common to $Z_i$ and $Z_j$.
  \item Orientation: Any directed line that can be obtained by translating $L$ and that
    intersects $Z_i$ and $Z_j$ first meets $Z_i$.
  \end{enumerate}
\end{theorem}

\begin{proof}
  One direction is trivial: if $C$ is a restriction of a power diagram to $\T$ then by Aurenhammer's original theorem the unrestricted power diagram has such a point set which satisfies orthogonality and orientation on all of $\V$ and thus also on $\T$.
  For the other direction,
Borgwardt and Frongillo~\citeyear{borgwardt2019power} show that a tiling of $\T$ is a restriction of a power diagram if and only if the following LP is feasible:

\vspace{\abovedisplayskip}
\noindent
\begin{tabular*}{\linewidth}{@{} l @{\extracolsep{\fill}} c @{\extracolsep{\fill}} r @{}}
  &
  $\begin{array}{rclcl}
     \lambda_{ij} \cdot a_{ij}  & = & p_i-p_j &  & \,\,\,\, \forall i \leq k, \forall j \in J_i\\
     \lambda_{ij}\cdot \gamma_{ij} & = & \gamma_j- \gamma_i & &   \,\,\,\, \forall i \leq k, \forall j \in J_i\\
     \lambda_{ij} & > & 0 & &  \,\,\,\, \forall i \leq k, \forall j \in J_i
   \end{array}$
                                                 &
\end{tabular*}
\vspace{\belowdisplayskip}

\noindent
Here $i,j \in \{1,\ldots,n\}$ index into the cells, $J_i$ is the indices of cells adjacent to cell $i$, and the cells are given by constants $a_{ij} \in \reals^d$, $\gamma_{ij}\in\reals$, so that the $i$th cell is defined as $\{t\in\T \mid a_{ij} \cdot t \leq \gamma_{ij} \; \forall j\in J_i\}$.
Thus, the variables are the sites $\{p_1,\ldots,p_n\} \subseteq \reals^d$ and pseudo-weights $\{\gamma_1,\ldots,\gamma_n\} \subseteq \reals$ which are in bijection with the true weights.

In our setting, by contrast, we are given the $p_i$ along with the $a_{ij}$ and $\gamma_{ij}$, and need only find real numbers $\gamma_i$ and $\lambda_{ij}$ for which the program is feasible. By orthogonality, there is a unique choice of $\lambda_{ij}$ satisfying the first constraint.  Furthermore, by orientation it is strictly positive, satisfying the third constraint.
For the second constraint, for any pair $(i,j)$ with $j\in J_i$ and choice of $\gamma_i$ there is a unique $\gamma_j$ satisfying the constraint.  Establishing the existence of a globally consistent set of choices is the heart of Aurenhammer's argument.  In particular, he shows that if $i$, $j$, $k$ share a vertex of $C$ then for arbitrary $\gamma_i$ the unique choices of $\gamma_j$ and $\gamma_k$ which satisfy the second constraint for $(i,j)$ and $(i,k)$ also satisfy it for $(j,k)$.  Global consistency then follows by a simple inductive construction.  Start by choosing a cell $i$ and arbitrary $\gamma_i$.  At each step we assign some $\gamma_j$.  If there is a $j$ for which $\gamma_j$ is unassigned and $j$ has a vertex which is shared with two assigned cells then, per Aurenhammer's argument, we can assign $\gamma_j$ consistent with all cells assigned so far.  Otherwise every unassigned cell has a most a single face in common with a single assigned cell and therefore we can choose one which has such a face and it can trivially be assigned consistently.
\end{proof}

\begin{theorem}
\label{thm:saksyu}
A tiling $C$ is a power diagram with sites $\{p_1,\ldots,p_n\}$ if and only if for all $t \in Z_i$ and $t' \in Z_j$ we have
$p_i \cdot (t' - t) \leq p_j \cdot (t'-t)$ (i.e. $C$ satisfies WMON).
\end{theorem}
\begin{proof}
If $C$ is a power diagram, then by definition
\begin{align*}
2 p_i \cdot t - w_i &\geq 2 p_j \cdot t - w_j\\
2 p_j \cdot t' - w_j &\geq 2 p_i \cdot t' - w_i.
\end{align*}
Adding these shows $C$ satisfies WMON.

Now suppose $C$ satisfies WMON. We show orthogonality and orientation.
For orthogonality, let $t,t' \in Z_i \cap Z_j$.  Then
$p_i \cdot (t' - t) = p_j \cdot (t'-t)$, or $(p_i-p_j) \cdot (t' - t) = 0$.
Thus, the face is orthogonal to $L$.
For orientation, let $t \in Z_i$ and $t' \in Z_j$ be on such a translated $L$.
That is, we can write $t' = t + c (p_j - p_i)$ for some $c \in \reals$.
By WMON,
$(p_j - p_i) \cdot (t'-t) \geq 0$, or $c (p_j - p_i) \cdot (p_j - p_i) \geq 0$.  Thus $c \geq 0$.  Therefore such a translated $L$ first meets $Z_i$.
\end{proof}

\begin{corollary}[\cite{saks2005weak}]
\label{cor:saks-yu}
If $\X$ is finite, $\T$ is convex, and $a$ satisfies WMON, then $a$ is implementable.
\end{corollary}
\begin{proof}
In order to apply Theorems~\ref{thm:power-diag} and~\ref{thm:saksyu}, it remains to show that an allocation rule $a$ satisfying WMON further implies that it defines a tiling.  This follows by a straightforward geometric argument that has been used in a number of previous proofs (see, e.g., Lemma 4.2 of~\cite{archer2008truthful}).  For completeness, provide it here.

Let $x \in \X$ be given.  We can define a polyhedron $P_x$ associated with $x$ by the intersection of the constraints $t \cdot (x-y) \geq  \inf_{t' s.t. a(t')=x} t' \cdot (x-y)$ for all $y \in X$.  By WMON, $\inf_{t' s.t. a(t')=x} t' \cdot (x-y) \geq \sup_{t s.t. a(t)=y} t \cdot (x-y)$, so for any distinct $x,y \in \X$, $P_x$ and $P_y$ are separated by the hyperplane  $t \cdot (x-y) =  \inf_{t' s.t. a(t')=x} t' \cdot (x-y)$, so their (relative) interiors are disjoint.  By construction $t \in P_{a(t)}$ for all $t$, which implies that these polyhedra cover $\T$.
\end{proof}

\section{Discussion}
\label{sec:discussion}

We have presented a model of truthful elicitation which generalizes and extends both mechanisms and scoring rules.  On the mechanism design side, we have seen how our framework provides simpler, more general, or more constructive proofs of a number of known results about implementability and revenue equivalence, some of which lead to new results about scoring rules.  On the scoring rules side, we have provided the first characterization for scoring rules for non-convex sets of probability distributions.  We have also extended our model to eliciting a property of the agent's private information.  This has been studied for specific cases in the scoring rules literature, but we have provided the first general characterization.  We also show how results about power diagrams in the scoring rules literature lead to a new proof of the Saks-Yu result in mechanism design.

Our analysis makes use of the fact that $\AS{t',t}$ is affine in $t$ to ensure that $G(t) = \sup_{t'} \AS{t',t}$ is a convex function.  However, this property continues to hold if $\AS{t',t}$ is instead a convex function of $t$.  Thus, a natural direction for future work is to investigate characterizations of convex scores.  While mechanisms can always be represented as affine functions by taking the types to be functions from allocations to $\reals$, it may be more natural to treat the type as a parameter of a (convex) utility function.  While many such utility functions are affine (e.g. dot-product valuations), others such as Cobb-Douglas functions are not.  Berger, M\"uller, and Naeemi~\shortcite{berger2009characterizing,berger2010path} have investigated such functions and given characterizations that suggest a more general result is possible.  Another potential application is scoring rules for alternate representations of uncertainty, several of which result in a decision maker optimizing a convex function~\cite{halpern2003reasoning}.

In one sense getting such a characterization is straightforward.  In the affine case we want $\AS{t',t}$ to be an affine function such that $\AS{t',t} \leq G(t)$ and $\AS{t',t'} = G(t')$.  Since we have fixed its value at a point, the only freedom we have is in the linear part of the function, and being such a linear function is exactly the definition of a subgradient.  So while our characterization of \scorenames is in some sense vacuous, it is also powerful in that it allows us to make use of the tools of convex analysis.  A similarly vacuous characterization is possible for the convex case: $\AS{t',t}$ is a convex function such that $\AS{t',t} \leq G(t)$ and $\AS{t',t'} = G(t')$.  The challenge is to find a way to state it that is useful and naturally handles constraints such as those imposed by the form of a utility function.

Many questions in the literature on properties remain open.  Most notable is the characterization of elicitable nonlinear and multidimensional properties --- the single dimensional case is covered in~\cite{lambert2008eliciting} and the linear vector-valued case in~\cite{abernethy2012characterization}.  We hope that the results and intuition from Section~\ref{sec:properties} will yield a useful characterization in this case.  Subsequent to this work we have made some partial progress on this approach~\cite{frongillo2015vector}. Another interesting direction is for non-functional properties: aside from the finite $R$ case, all work in the literature to our knowledge assumes that $\Gamma$ is a function (having a single correct report for each type).  The generality of Theorem~\ref{thm:prop-char} may prove useful in exploring non-functional settings as well.  A result requiring few regularity conditions on $\Gamma$ would be useful in domains such as statistics where natural properties like the median cannot in general be expressed as functions.

 Theorem~\ref{thm:power-diag} shows that scoring rules for finite properties are essentially equivalent the weights and points that induce a power diagram.  As power diagrams are known to be connected to the spines of amoebas in algebraic geometry, aspects of toric geometry used by string theorists, and tropical hypersurfaces in tropical geometry~\cite{van_manen2005power}, there may be useful characterization results in those fields as well.  The last is particularly suggestive given the recent use of tropical geometry techniques in mechanism design~\cite{baldwin2012tropical}.

While our examples have focused on mechanism design and scoring rules, another interesting direction to pursue is other settings where our results may be applicable.  One natural domain
is the literature on M-estimators in machine learning, statistics and economics.  Essentially, this literature takes a loss function (i.e. a scoring rule) and asks what
it elicits.  For example, the mean is an M-estimator induced by the squared error loss function.  Some work in this literature (e.g.~\cite{negahban2010unified}) requires that the loss function satisfy certain conditions, and our results may be useful in characterizing and supplying such loss functions.

\appendix

\section{Convex Analysis Primer}
\label{sec:cvx-primer}

In this appendix, we review some facts from convex analysis that are used in the paper.

\begin{fact}
Let $\{f_t \in \affine(\V\to\reals)\}_{t\in\T}$ be a parameterized family of affine functions.  Then $G(t) = \sup_{t' \in T} f_{t'}(t)$ is convex as the pointwise supremum of convex functions.
\end{fact}

This follows because convex functions are those with convex epigraphs.  The epigraph of this supremum is the intersection of the epigraphs of the individual functions, which is a convex set as the intersection of convex sets.

\begin{fact}
$d : \reals \to \reals $ is a selection of subgradients of a convex function on $\reals$ if and only if it is monotone non-decreasing.
\end{fact}

See~\cite[Theorem 24.3]{rockafellar1997convex} for a proof of a slightly more general statement.

\begin{fact}
For convex $G$ on convex $\T$, $\{dG_t \in \linear(\V\to\reals)\}_{t\in\T} \selsubgrad{G}{\T}$ satisfies path independence.
\end{fact}

Informally, path independence means that integrals of $dG_t$ do not depend on the path through $\T$ chosen (see Equation~\eqref{eq:lineintegral} and following for a formal definition).  Since the restriction of $G$ to a line is a one-dimensional convex function, $G(y) - G(x) = \int_{L_{xy}} dG_t(y-x) dt$~\cite[Corollary 24.2.1]{rockafellar1997convex}.  Summing along the individual lines in a path from $x$ to $y$ gives that the value of the path integral is $G(y) - G(x)$ regardless of the path chosen.

The following is a classic result in convex analysis (cf.~\cite[Thm E.1.4.1]{urruty2001fundamentals}) which we prove for completeness.

\begin{fact}\label{lem:properties-dual-optimal}
  Let $G:\V\to\extreals$ be convex.  
  Then for all $v\in\V, d\in\V^*$,
  \begin{equation*}
    \label{eq:properties-4}
    G^*(d) = d(v) - G(v) \, \iff \, d \in \partial G_v.
  \end{equation*}
\end{fact}
\begin{proof}
  We can simply break down the conditions step by step:
  \begin{align*}
    G^*(d) = d(v) - G(v)
    &\iff v \in \argsup\nolimits_{v'\in \V} \; d(v') - G(v') \\
    &\iff \forall v'\in\V, \; d(v) - G(v) \geq d(v') - G(v')\\
    &\iff \forall v'\in\V, \; G(v') \geq G(v) + d(v'-v),
  \end{align*}
  where in the last step we merely negated and added $d(v')\in\reals$ to both sides.
\end{proof}

\begin{fact}
For any convex function $G$, the set $\partial G^{-1}(d) \defeq \{x \in \dom(G) : d\in\partial G_x\}$ is convex.   
\end{fact}

\begin{proof}
  Let $x,x'\in\partial G^{-1}(d)$; then one easily shows (cf. Lemma~\ref{lem:subgradients}) that $G(x) - G(x') = d(x-x')$.  Now let $\hat x = \alpha x + (1-\alpha) x'$; we have,
  \begin{align}
    G(\hat x)
    &\leq \alpha G(x) + (1-\alpha) G(x') \label{eq:properties-10}\\
    &= \alpha (G(x) - G(x')) + G(x') \nonumber \\
    &= \alpha d(x-x') + G(x') \nonumber \\
    &= d(\hat x-x') + G(x') \label{eq:properties-11} \\
    &\leq G(\hat x), \label{eq:properties-12}
  \end{align}
  where we applied convexity of $G$ in~\eqref{eq:properties-10} and the subgradient inequality for $d$ at $x'$ in~\eqref{eq:properties-12}.  Hence, by eq.~\eqref{eq:properties-11} we have shown $G(\hat x) - G(x') = d(\hat x - x')$, so by Lemma~\ref{lem:subgradients}, $d \in \partial G_{\hat x}$.
\end{proof}

\section{Characterizing Truthful Mechanisms}
\label{sec:convex}

While our theorem provides a characterization of truthful mechanisms in terms of convex consumer surplus functions, this is not always the most natural representation for a mechanism.  In this section, we examine two other approaches to characterizing truthful mechanisms that have been explored in the literature and show that they have insightful interpretations in convex analysis, which allows us to greatly simplify their proofs.  Furthermore, our phrasing of these results is as conditions for a parameterized family of linear functions to be a selection of subgradients of a convex function.  We believe this phrasing converts known results in mechanism design into new results in convex analysis.  It also shows how any such result in convex analysis would give a characterization of implementable mechanisms.  Note that certain results in this section require an assumption that the relevant parameterized families are in fact real-valued, which is natural given our focus on mechanism design.

\subsection{Subgradient characterizations}
\label{sec:subgradient}

From an algorithmic perspective, it may be more natural to focus on the design of the allocation rule $f$.  There is a large literature that focuses on when there exists a choice of payments $p$ to make $f$ into a truthful mechanism (e.g. \cite{saks2005weak,ashlagi2010monotonicity}).  Viewed through our theorem, this becomes a very natural convex analysis question: when is a function $f$ a subgradient of a convex function?
\footnote{More precisely, we want for all $t$ the allocation $f(t)$ to be a subgradient at t.  Equivalently, we can view $f$ as a parameterized family of functions, which is how we state our results.}
Unsurprisingly, the central result in the literature is closely connected to convex analysis.  

\begin{definition}
  A family $\{d_t \in \linear(\V\to\reals)\}_{t\in\T}$ satisfies \emph{cyclic monotonicity (CMON)} if for all finite sets $\{t_0,\ldots,t_k\} \subseteq \T$,
  \begin{equation}
    \label{eq:cmon}
    \sum_{i=0}^k d_{t_i} (t_{i+1}-t_i) \leq 0,
  \end{equation}
  where indices are taken modulo $k+1$.  The weaker condition that~\eqref{eq:cmon} hold for all pairs $\{t_0,t_1\}$ is known as \emph{weak monotonicity (WMON)}.
\end{definition}

A well known characterization from convex analysis is that a function $f$ defined on a convex set is a subgradient of a convex function on that set iff it satisfies CMON~\cite{rockafellar1997convex}.  Rochet's~\shortcite{rochet1987necessary} proof that payments exist to implement $f$ on a possibly non-convex $\T$ iff $f$ satisfies CMON is effectively a proof of a generalization of this theorem.  Rochet notes that his proof is adapted from the one given in Rockafellar's text~\shortcite{rockafellar1997convex} of the weaker theorem where $\T$ is restricted to be convex.  We adapt Rochet's proof to highlight how its core is a construction of $G$.  As we use this basic construction several times, we first analyze it independly.

Given any family $\{d_t\}_{t\in\T}$ of linear functions in $\linear(\V\to\reals)$, define $P_d:\T\times\V\to\extreals$ as follows:
\footnote{Note that the second argument of $P_d$ is from $\V$ rather than $T \subset \V$ because we wish to apply this when, e.g., $t' \in \conv(\T)$.}
\begin{equation}
  \label{eq:p-d}
  P_d(t,t')
  \defeq \;
  \sup_{\substack{k\in\mathbb{N},\; \{t_1, \ldots, t_k\} \subseteq \T\\ t_0=t,\; t_{k+1} = t'}}
  \;\;\sum_{i = 0}^k d_{t_i}(t_{i+1}-t_i).
\end{equation}
One way to interpret  $P_d(t,t')$ is as the length of the shortest path from $t$ to $t'$ in a graph with edge weights determined by $-d$, and in that form has seen extensive use in mechanism design~\cite{vohra2011mechanism}.  We interpret it somewhat differently, as the best lower bound on $G(t') - G(t)$ for an arbitrary convex function $G$ with subgradients $d$ (and infinity if there is no such convex function).  In particular, computing the best lower bound at every point yields a convex function. 

\begin{lemma}
\label{lem:affine-1}
  Let $\{d_t \in \linear(\V\to\reals)\}_{t\in\T}$ be given.
  If $d$ satisfies CMON, then for all $t,t'\in\T$ and all $t''\in\V$, the following hold:
  \begin{enumerate}
  \item $P_d(t,t') + P_d(t',t'') \leq P_d(t,t'')$
  \item $d_t(t''-t) \leq P_d(t,t'')$
  \item $P_d(t,t) = 0$
  \item $P_d(t,t') + P_d(t',t) \leq 0$
  \item $P_d(t,\cdot)$ is convex and real-valued on $\convhull(\T)$, with $d \in \partial P_d(t,\cdot)$ on $\T$
  \end{enumerate}
  Otherwise, $P_d \equiv \infty$ on all inputs.
\end{lemma}

\begin{proof}
If CMON is not satisfied, then there is a cycle $C = t_0, \ldots, t_k, t_0$ with positive sum.  Then for any $t$ and $t'$ the path $t C^j t'$ that consists of starting at $t$, going to $t_0$, going around the cycle $j$ times, then going to $t'$ has a sum that goes to infinity as $j$ goes to infinity.  For the remainder, assume that CMON is satisfied.

  \begin{enumerate}
  \item $P_d(t,t') + P_d(t',t'') \leq P_d(t,t'')$

\begin{align*}
&P_d(t,t') + P_d(t',t'')\\
&=\sup_{\substack{k\in\mathbb{N},\; \{t_1, \ldots, t_k\} \subseteq \T\\ t_0=t,\; t_{k+1} = t'}}
  \;\;\sum_{i = 0}^k d_{t_i}(t_{i+1}-t_i)
+\sup_{\substack{k\in\mathbb{N},\; \{t_1, \ldots, t_k\} \subseteq \T\\ t_0=t',\; t_{k+1} = t''}}
  \;\;\sum_{i = 0}^k d_{t_i}(t_{i+1}-t_i)\\
&=\sup_{\substack{j,k\in\mathbb{N},\; \{t_1, \ldots, t_k\} \subseteq \T\\ t_0=t,\; t_j = t', \; t_{k+1} = t''}}
  \;\;\sum_{i = 0}^k d_{t_i}(t_{i+1}-t_i)\\
&\leq\sup_{\substack{k\in\mathbb{N},\; \{t_1, \ldots, t_k\} \subseteq \T\\ t_0=t,\; t_{k+1} = t''}}
  \;\;\sum_{i = 0}^k d_{t_i}(t_{i+1}-t_i)\\
&=P_d(t,t'')
\end{align*}

  \item $d_t(t''-t) \leq P_d(t,t'')$

Taking $k = 0$ shows that $d_t(t''-t)$ is an element of set over which the supremum is taken.

  \item $P_d(t,t) = 0$

By CMON, $P_d(t,t) \leq 0$.  By claim (2), $d_t(t-t) = 0 \leq P_d(t,t)$.

  \item $P_d(t,t') + P_d(t',t) \leq 0$

By claims (1) and (3), $P_d(t,t') + P_d(t',t) \leq P(t,t) = 0$.

  \item $P_d(t,\cdot)$ is convex and real-valued on $\convhull(\T)$, with $d \in \partial P_d(t,\cdot)$ on $\T$

By CMON, for $t' \in \T$ $P_d(t,t') \leq -d_t(t_0 - t')$.
Thus, $P_d(t,t')$ is finite on $\T$.
$P_d(t,\cdot)$ is a pointwise supremum of convex functions, so is convex.
By convexity, it is also finite on $\convhull(\T)$.
For any $t' \in \T$ and $t'' \in \convhull(\T)$,
\begin{align*}
P_d(t,t') + d_{t'}(t'' -t')
&= d_{t'}(t''-t') +
 \sup_{\substack{k\in\mathbb{N},\; \{t_1, \ldots, t_k\} \subseteq \T\\ t_0=t,\; t_{k+1} = t'}}
 \sum_{i = 0}^k d_{t_i}(t_{i+1}-t_i)\\
&=  \sup_{\substack{k\in\mathbb{N},\; \{t_1, \ldots, t_k\} \subseteq \T\\ t_0=t,\; t_k = t'\;t_{k+1} = t''}} \sum_{i = 0}^{k} d_{t_i}(t_{i+1}-t_i)\\
&\leq  \sup_{\substack{k\in\mathbb{N},\; \{t_1, \ldots, t_k\} \subseteq \T\\ t_0=t,\; t_{k+1} = t''}} \sum_{i = 0}^k d_{t_i}(t_{i+1}-t_i)\\
& = P(t,t''),
\end{align*}
so $d_t$ satisfies~\eqref{eq:subgradient}.

  \end{enumerate}
\end{proof}

Having extracted the construction at the core of Rochet's proof, the rephrasing of his result as a statement about convex functions now follows easily.
\begin{theorem}[Adapted from Rochet~\shortcite{rochet1987necessary}]
\label{thm:cmon}
A family $\{d_t \in \linear(\V\to\reals)\}_{t\in\T}$ satisfies CMON if and only if there exists a convex  $G : \convhull(\T)\to\reals$ such that $\{d_t\}_{t\in\T} \selsubgrad{G}{\T}$.
\end{theorem}
\begin{proof}
Given such a $G$, by \eqref{eq:subgradient} we have $d_{t_i}(t_{i+1} - t_i) \leq G(t_{i+1}) - G(t_i)$.  Summing gives~\eqref{eq:cmon}.
Given such a family $\{d_t\}_{t \in \T}$, fix some $t_0 \in \T$ and set $G:t\mapsto P_d(t_0,t)$.  The result follows from Lemma~\ref{lem:affine-1}(5).
\end{proof}

A number of papers have sought simpler and more natural conditions than CMON that are necessary and sufficient in special cases, e.g.~\cite{saks2005weak,archer2008truthful,ashlagi2010monotonicity}.  These results are typically proven by showing they are equivalent to CMON.  However, it is much more natural to directly construct the relevant $G$.  As an example, we show one such result has a simple proof using our framework.  This particular proof also has the advantage of providing a characterization of the payments that is more intuitive than the supremum in Rochet's construction.

 As in Myerson's~\shortcite{myerson1981optimal} construction for the single-parameter case, we construct a $G$ by integrating over $d_t$.  In particular, for any two types $x$ and $y$ our construction makes use of the line integral
\begin{equation}
\label{eq:lineintegral}
\int_{L_{xy}} d_t(y-x) dt = \int_0^1 d_{(1-t)x + ty}(y-x) dt.
\end{equation}
As Berger et al.~\shortcite{berger2009characterizing} and Ashlagi et al.~\shortcite{ashlagi2010monotonicity} observed, if $\{d_t\}_{t \in \T}$ satisfies WMON and $\T$ is convex, this (Riemann) integral is well defined because it is the integral of a monotone function.  If these line integrals vanish around all triangles (equivalently $\int_{L_{xy}} d_t(y-x) dt + \int_{L_{yz}} d_t(z-y) dt = \int_{L_{xz}} d_t(z-x) dt)$) we say $\{d_t\}$ satisfies {\em path independence}.

\begin{theorem}[adapted from~\cite{muller2007weak}]\label{thm:wmon}
For convex $\T$, a family $\{d_t \in \linear(\V\to\reals)\}_{t\in\T}$ is a selection of subgradients of a convex function if and only if $\{d_t\}_{t \in \T}$ satisfies WMON and path independence.
\end{theorem}

\begin{proof}
Given a convex function $G$ and selection of subgradients $\{d_t\}$, $\{d_t\}$ satisfies CMON and thus WMON.  Path independence also follows from convexity (Rockafellar~\shortcite{rockafellar1997convex} p. 232).
Now given a $\{d_t\}$ that satisfies WMON and path independence, fix a type $t_0 \in \T$ and define
$G(t') = \int_{L_{t_0t'}} d_t(t' - t_0) dt$ (well defined by WMON as the integral of a monotone function).  Given $x,y,z \in \T$ such that $z = \lambda x + (1 - \lambda) y$, by path independence and the linearity of $d_z$ we have
\begin{align*}
&\lambda G(x) + (1 - \lambda) G(y)\\
&= G(z)+ \lambda \int_{L_{zx}} d_t(x - z)dt + (1- \lambda) \int_{L_{zy}} d_t(y - z)dt\\
&\geq G(z) + \lambda d_z(x - z) + (1 - \lambda) d_z(y - z) = G(z),
\end{align*}
so $G$ is convex.  Similarly, for $x, y \in \T$, $d_t$ satisfies~\eqref{eq:subgradient} because
$$d_x(y - x) \leq \int_{L_{xy}} d_t(y-x)dt = G(y) - G(x). \qedhere$$
\end{proof}

\subsection{Local Characterizations}
\label{sec:local}

In many settings, it is easier to reason about the behavior of a mechanism given small changes to its input rather than arbitrary changes, so several authors have sought to characterize truthful mechanisms using local conditions~\cite{archer2008truthful,berger2009characterizing,carroll2012when}.  
We show in this section how many of these results are in essence a consequence of a more fundamental statement, that convexity is an inherently local property.  For example, in the twice differentiable case it can be verified by determining whether the Hessian is positive semidefinite at each point.  We start with a local convexity result, and use it to show that an \scorename is truthful if and only if it satisfies a very weak local truthfulness property introduced by Carroll~\shortcite{carroll2012when}.  Afterwards we turn to a characterization by Archer and Kleinberg~\shortcite{archer2008truthful} that proved a similar theorem for a different notion of local truthfulness.  Our results (specifically Theorem~\ref{thm:wlsg}) show that these two notions of local truthfulness are equivalent because Archer and Kleinberg's definition corresponds to the property of being a local subgradient, while Carroll's corresponds to the property of being a weak, local subgradient, which we now define.

\begin{definition}
  Let $\T$ be convex. A family $\{d_t \in \linear(\V\to\extreals)\}_{t \in \T}$ is a \emph{weak local subgradient (WLSG)} of a convex function $G:\T\to\extreals$ if for all $t\in\T$ there exists an open neighborhood $U_t$ of $t$ such that for all $t'\in U_t$,
  \begin{equation}
    \label{eq:wlsg}
    G(t) \geq G(t') + d_{t'}(t-t') \quad \text{ and } \quad G(t') \geq G(t) + d_{t}(t'-t).
  \end{equation}
  Furthermore, if for every $s\in\T$, eq.~\eqref{eq:wlsg} holds for all $t,t'\in U_s$, we say $\{d_t\}_{t \in \T}$ is a \emph{local subgradient (LSG)} of $G$.
\end{definition}

We now show that being a WLSG is a necessary and sufficient condition for a family of functions to be a selection of subgradients.  The proof is heavily inspired by Carroll~\shortcite{carroll2012when}.
\begin{theorem}\label{thm:wlsg}
Let $\T$ be convex.
  A family $\{d_t \in \linear(\V \to \extreals )\}_{t \in \T}$ is a selection of subgradients of a convex function $G:\T\to\extreals$ if and only if it is a WLSG of $G$.
\end{theorem}
\begin{proof}[(Adapted from Carrol~\citeyear{carroll2012when})]
  As usual, the forward direction is trivial.  For the other, let $t,t'\in\T$ be given; we show that the subgradient inequality for $d_{t'}$ holds at $t$.  By compactness of $\convhull(\{t,t'\})$, we have a finite set $t_i = \alpha_i t' + (1-\alpha_i) t$, where $0 = \alpha_0 \leq \cdots \leq \alpha_{k+1} = 1$, such that WLSG holds between each $t_i$ and $t_{i+1}$.
(The cover $\{U_s \,|\, s \in \convhull(\{t,t'\})\}$ has a finite subcover.  Take $t_{2j}$ from the subcover and $t_{2j+1} \in U_{t_{2j}} \cap U_{t_{2j+2}}$.)
By the WLSG condition~\eqref{eq:wlsg}, we have for each $i$,
  \begin{align}
    0 \geq \; & G(t_{i+1}) - G(t_i) + d_{t_{i+1}}(t_i-t_{i+1}) \label{eq:borkybork1}
    \\
    0 \geq \; & G(t_i) - G(t_{i+1}) + d_{t_i}(t_{i+1}-t_i). \label{eq:borkybork2}
  \end{align}
  Now using the identity $t_{i+1} - t_i = (\alpha_{i+1} - \alpha_i)(t'-t)$ and adding $\alpha_i/(\alpha_{i+1} - \alpha_i)$ times~\eqref{eq:borkybork1} to $\alpha_{i+1}/(\alpha_{i+1} - \alpha_i)$ times~\eqref{eq:borkybork2}, we have
  \begin{equation}
    \label{eq:borkyborkybork}
    0 \geq G(t_i) - G(t_{i+1}) + \alpha_{i+1} d_{t_i}(t'-t) - \alpha_i d_{t_{i+1}}(t'-t).
  \end{equation}
  Summing~\eqref{eq:borkyborkybork} over $0\leq i\leq k$ gives
  \begin{align*}
    0 & \geq G(t_0) - G(t_{k+1}) + \alpha_{k+1} d_{t_0}(t'-t) - \alpha_0 d_{t_{k+1}}(t'-t),
  \end{align*}
  which when recalling our definitions for $\alpha_i$ and $t_i$ yields the result.
\end{proof}

The WLSG condition translates to an analogous notion in terms of truthfulness, \emph{weak local truthfulness}.
\begin{definition}
  An \scorename is \emph{weakly locally truthful} if for all $t\in\T$ there exists some open neighborhood $U_t$ of $t$, such that truthfulness holds between $t$ and every $t'\in U_t$, and vice versa.  That is,
  \begin{equation}
    \label{eq:weak-local-truthful}
    \forall t\in \T,\; \forall t' \in U_t,\;\; \AS{t',t} \leq \AS{t,t} \;\text{ and }\; \AS{t,t'} \leq \AS{t',t'}.
  \end{equation}
\end{definition}

\begin{corollary}[Generalization of Carroll~\shortcite{carroll2012when}]\label{cor:wlt}
  An \scorename $\ASdef$ for convex $\T$ is truthful if and only if it is weakly locally truthful.
\end{corollary}
\begin{proof}
  Defining $G(t) := \AS{t,t}$, by weak local truthfulness we may write
  \begin{align*}
    G(t) = \AS{t,t} &\geq \AS{t',t} = G(t') + \AS[_\ell]{t',t-t'}\\
    G(t') = \AS{t',t'} &\geq \AS{t,t'} = G(t) + \AS[_\ell]{t,t'-t},
  \end{align*}
  where $t'$ is local to $t$ and $\AS[_\ell]{t,\cdot}$ is the linear part of $\AS{t,\cdot}$.  This says that $d_t = \AS[_\ell]{t,\cdot}$ satisfies WLSG for convex function $G$; the rest follows from Theorem~\ref{thm:wlsg} and Theorem~\ref{thm:main-char}.
\end{proof}

Finally, in the spirit of Section~\ref{sec:subgradient}, Archer and Kleinberg~\shortcite{archer2008truthful} characterized local conditions under which an allocation rule can be made truthful.  A key condition from their paper is \emph{vortex-freeness}, which is a condition they show to be equivalent to local path independence (analagous to our terminology of weak local subgradients it can be thought of as weak local path independence).  The other condition, local WMON, means that WMON holds in some neighborhood around each type.  Their result then follows from the observation that local WMON and local path independence imply local subgradient.  While this particular proof is not significantly simpler than the original, we believe it is somewhat more natural and clarifies the connection between the underlying reasons a notion of local truthfulness suffices both here and in Carroll's setting.

\begin{corollary}\label{cor:lwmon}
Let $\T$ be convex.  A family $\{d_t \in \linear(\V \to \reals)\}_{t \in \T}$ is a selection of subgradients of a convex function if and only if it satisfies local WMON and is vortex-free.
\end{corollary}

\begin{proof}
  
  We prove the reverse direction; suppose $\{d_t\}_{t \in \T}$ satisfies local WMON and is vortex-free.  From Lemma 3.5 of Archer and Kleinberg~\citeyear{archer2008truthful} we have that vortex-freeness is equivalent to path independence, so by Theorem~\ref{thm:wmon} for all $t$ there exists some open $U_t$ such that $\{d_{t'}\}_{t' \in U_t}$ is the subgradient of some convex function $G^{(t)}:U_t\to\reals$.  We need only show the existence of some $G$ such that $\{d_t\}_{t \in \T}$ is the subgradient of $G$ on each $U_t$; the rest follows from Theorem~\ref{thm:wlsg}.

Fix some $t_0\in\T$ and define $G(t) = \int_{L_{t_0\,t}} d_{t'} dt'$, which is well defined by compactness of $\convhull(\{t_0,t\})$ and the fact that a locally increasing real-valued function is increasing.  But for each $t'$ and $t\in U_{t'}$ we can also write $G^{(t')}(t) = \int_{L_{t'\,t}} d_{t''} dt''$ by~\cite[p. 232]{rockafellar1997convex}, and now by path independence we see that $G$ and $G^{(t')}$ differ by a constant.  Hence $\{d_t\}_{t \in \T}$ must be a subgradient of $G$ on $U_{t'}$ as well, for all $t'\in\T$.
\end{proof}

\section{Revenue Equivalence}
\label{sec:rev-eq}

Perhaps the most celebrated result in auction theory is the revenue equivalence theorem, which states that, in a single item auction, the revenue from an agent (equivalently that agent's consumer surplus) is determined up to a constant by the equilibrium probability that each possible type of that agent will receive the item~\cite{myerson1981optimal}.  A large body of work has looked for more general conditions under which this property holds (see, e.g., \cite{krishna2001convex}) or what can be said when it does not~\cite{carbajal2012mechanism}.  One general approach is due to Heydenreich et al.~\shortcite{heydenreich2009characterization}, who use a graphical representation related to CMON.  Given our main theorem, this is unsurprising.  In convex analysis terms, asking whether an implementable allocation rule satisfies revenue equivalence is asking whether all convex functions that have a selection of their subgradients that corresponds to that allocation rule are the same up to a constant.  As we saw in the proof of Lemma~\ref{lem:affine-1}, CMON permits the natural construction of a convex function from its subgradient via \eqref{eq:cmon}.  Intuitively, if we know the payments we want for some subset of types, we can check if those are consistent with a desired payment for some other type by checking whether this construction still works, both in terms of the constraints of the existing types on the new one and the new one on the existing ones.  The following theorem applies this insight to get a result that is stronger than revenue equivalence as iteratively applying it characterizes the possible payments for {\em every} mechanism.

\begin{theorem}
\label{thm:rev}
Let $G$ be a convex function on $\convhull(\T)$, $d=\{d_t\}_{t \in \T}$ a selection of its subgradients on $\T$, $S \subseteq \T$ non-empty, $t^* \in \T \setminus S$, and $c$ be given.  Then there exists a convex $G'$ on $\convhull(\T)$ agreeing with $G$ on $S$, with $\{d_t\}_{t \in \T} \selsubgrad{G'}{\T}$ and $G'(t^*) = c$, if and only if
\begin{equation}
\label{eqn:rev-bound-new}
\sup_{t_0\in S} G(t_0) + P_d(t_0,t^*)
\leq c \leq
\inf_{t_0\in S} G(t_0) - P_d(t^*,t_0)
\end{equation}
\end{theorem}
\begin{proof}
Given such a $G'$, the LHS of~\eqref{eqn:rev-bound-new} becomes $\sup_{t_0\in S} G'(t_0) + P_d(t_0,t^*) \leq G'(t^*)$.  Applying the definition of $P_d$~\eqref{eq:p-d} and then repeatedly applying the subgradient inequality~\eqref{eq:subgradient} yields the desired inequality.  Similarly, the RHS of~\eqref{eqn:rev-bound-new} can be rewritten as $G'(t^*) + P_d(t^*,t_0) \leq G'(t_0)$ for all $t_0 \in S$, and the definition and subgradient inequality applied.

Now suppose \eqref{eqn:rev-bound-new} holds.  Let $G'(t) \defeq \max\left\{c+P_d(t^*,t),\; \sup_{t_0\in S} G(t_0) + P_d(t_0,t)\right\}$.  By Theorem~\ref{thm:cmon}, $d$ satisfies CMON, so by  Lemma~\ref{lem:affine-1} $G'$ is convex, finite-valued on $\convhull(\T)$, and has $\{d_t\}\in\partial G'$.  Hence, we need only show that $G'$ agrees with $G$ on $S$ and has $G'(t^*)=c$.

First, fixing any $t\in S$, we will establish the following:
\begin{equation}
  \label{eq:affine-g-sup-p-d}
  G(t) = \sup_{t_0\in S} G(t_0) + P_d(t_0,t).
\end{equation}
As $P_d(t,t)=0$ from Lemma~\ref{lem:affine-1}(3), we have $G(t) = G(t) + P_d(t,t) \leq \sup_{t_0\in S} G(t_0) + P_d(t_0,t)$.  Furthermore, $G(t_0) + P_d(t_0,t) \leq G(t)$ for all $t_0\in \T$ by repeated application of the subgradient inequality~\eqref{eq:subgradient}.  Hence, we have $\sup_{t_0\in S} G(t_0) + P_d(t_0,t) \leq G(t)$ as well.

By eq.~\eqref{eq:affine-g-sup-p-d}, we can write $G'(t) = \max\{c+P_d(t^*,t), G(t)\}$ when $t\in S$.  But by the RHS of eq.~\eqref{eqn:rev-bound-new}, we see $c+P_d(t^*,t) \leq G(t)$, so $G'(t) = G(t)$.  Similarly, applying the LHS of eq.~\eqref{eqn:rev-bound-new} and $P_d(t^*,t^*)=0$ to the definition of $G'(t^*)$, we have $G'(t^*) = c$.
\end{proof}

Viewed through Theorem~\ref{thm:rev}, revenue equivalence holds when the upper and lower bounds from \eqref{eqn:rev-bound-new} match after the value of $G$ is fixed at a single point.  This allows us to derive a necessary and sufficient condition for revenue equivalence that is equivalent to that given by Heydenreich et al.~\shortcite{heydenreich2009characterization} and actually applies to all \scorenames.  For example, this gives a revenue equivalence theorem for mechanisms with partial allocation.

\begin{corollary}[Revenue Equivalence]
  \label{cor:rev-eq}
Let a truthful \scorename $\ASdef$ be given, and $d = \{d_t\}_{t\in\T}$ be the corresponding selection of subgradients  from~\eqref{eq:main-char}.  Then every truthful \scorename $\ASdef[']$ with the same corresponding selection of subgradients differs from $\ASS$ by a constant (i.e. $\AS{t',t} = \AS[']{t',t} + c$) if and only if $P_d(t',t) + P_d(t,t') = 0$ for all $t,t'\in\T$.
\end{corollary}

\begin{proof}
We will show that the convex function $G$ from eq.~\eqref{eq:main-char} is unique up to a constant if and only if $P_d(t',t) + P_d(t,t') = 0$ for all $t,t'\in\T$.

For the forward direction, let $t_0\in \T$ be arbitrary.  Then for all $t\in\T$, taking~\eqref{eqn:rev-bound-new} with $S=\{t_0\}$ and $G(t) \defeq c + P_d(t_0,t)$ gives the condition $G(t_0) + P_d(t_0,t) \leq G'(t) \leq G(t_0) - P_d(t,t_0)$ for the value of $G'(t)$.  But as $P_d(t,t_0) = -P_d(t_0,t)$ we have $G'(t) = P_d(t_0,t) + G(t_0) = G(t)$ for all $t$.

For the reverse direction, assume $P_d(t^1,t^2) \neq -P_d(t^2,t^1)$ for some $t^1,t^2\in\T$, and let $G^1(t) \defeq P_d(t^1,t)$ and $G^2(t) \defeq P_d(t^1,t^2) + P_d(t^2,t)$.  We easily check from Lemma~\ref{lem:affine-1}(3) that $G^1(t^2) = G^2(t^2) = P_d(t^1,t^2)$, but we have $G^1(t^1) = 0$ while $G^2(t^1) = P_d(t^1,t^2) + P_d(t^2,t^1) \neq 0$.
\end{proof}

We note that these two results are similar to results of Kos and Messner~\shortcite{kos2012extremal}.  The main novelties in our version are showing that every value in the interval yields a convex function (as opposed to merely the extremal ones), the ability to characterize possible values after the values at multiple points are fixed (as opposed to a single point),
and the framing in terms of convex analysis.

The conditions given by Theorem~\ref{thm:rev} and Corollary~\ref{cor:rev-eq}, while general, are not particularly intuitive.  However, there are a number of special cases where they do have natural interpretations for mechanism design.  The first is when the set of types is finite.  In this setting (explored in an auction theory context in, e.g., \cite{diakonikolas2012efficiency}) it is well known that revenue equivalence does not hold.  The finite set of constraints~\eqref{eqn:rev-bound-new} can be used in general as a linear program to, e.g., maximize revenue (see Section~6.5.2 of~\cite{vohra2011mechanism} for an example).  In particular cases, they may become simple enough to have a nice characterization.  For example, in the single-parameter setting only a linear number of paths need be considered.  This setting is illustrated in Figure~\ref{fig:interp}.

\begin{figure}[!ht]
  \centering
  \begin{tabular}{ccc}
    \includegraphics[width=0.25\textwidth]{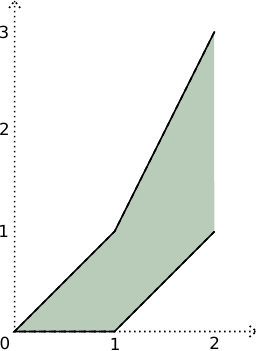} & \hspace{40pt} &
    \includegraphics[width=0.25\textwidth]{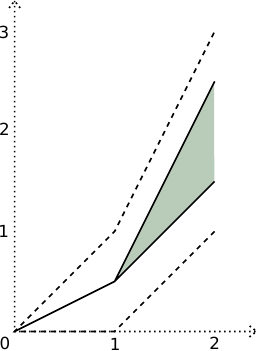} \\
    (a) && (b)
  \end{tabular}
  \caption{Consider a one-dimensional setting with type space $\T = \{0,1,2\}$ and $d_0  = 0, d_1 = 1, d_2 =2$.  In (a), we fix $G(0) =0$, yielding a range of possible values dictated by the subgradients: $0 \leq G(1) \leq 1$ and $1 \leq G(2) \leq 3$.  We can pick any point in the resulting set and fix G there.  However, we cannot pick any increasing function: in (b), we fix $G(1) = 0.5$, restricting $G(2)$ to the interval $[1.5,2.5]$.}
  \label{fig:interp}
\end{figure}

More broadly, as we saw in the proof of Theorem~\ref{thm:wmon}, the (supremum over the) sum can often be interpreted as an integral.  In particular, the fact that $G$ is convex guarantees that (under mild conditions) integrals of a selection of its subgradient are path independent and the integral from $t$ to $t'$ gives $G(t') - G(t)$.  If $\T$ is connected by smooth paths (e.g. if it is convex), this means that $\T$ satisfies revenue equivalence for all implementable mechanisms (previously shown under a somewhat different notion of the set of types~\cite{heydenreich2009characterization}).  As it is particularly simple to prove, we state the version for convex $\T$.
\begin{corollary}
\label{thm:old-rev-eq}
Let $\T$ be convex, a truthful \scorename $\ASdef$ be given, and $\{\subgrad{G}_t\}_{t\in\T}$ be the corresponding selection of subgradients  from~\eqref{eq:main-char}.  Then any truthful \scorename $\ASdef[']$ with the same corresponding selection of subgradients differs from $\ASS$ by a constant (i.e. $\AS{t',t} = \AS[']{t',t} + c$).
\end{corollary}

\begin{proof}
By Theorem~\ref{thm:main-char}, we know that $\ASS$ and $\ASS'$ only differ only in their choice of convex function $G$.  However, each choice has the same selection of subgradients, and two convex functions with the same selection of subgradents differ by a constant~\cite{rockafellar1997convex}.  For intuition, see the construction of $G$ by integrating its subgradients in the proof of Theorem~\ref{thm:wmon}.
\end{proof}

\section{Duality in elicitiation}
\label{sec:append-duality}

We now provide mathematical foundation for notions of duality from \sect~\ref{chap:propertiesduality-elicitation}.
To begin, we need a dual vector space with more structure than simply $\linear(\V\to\reals)$.
For this we use the notion of a \emph{dual pair}, which is a standard setting for convex analysis in infinite-dimensional spaces.

\begin{definition}[{\cite[\sect~5.14]{aliprantis2007infinite}}]
  \label{def:properties-dual-pair}
  A pair of topological vector spaces $(\V,\V^*)$ is a \emph{dual pair} if it is equipped with a bilinear form $\inprod{\cdot,\cdot}:\V\times\V^* \to \reals$ which separates points, in the sense that $\forall v^*\, \inprod{v,\cdot}\equiv 0$ implies $v=0$ and  $\forall v\, \inprod{\cdot,v^*}\equiv 0$ implies $v^*=0$.
\end{definition}

Note that the above assumption that $(\V,\V^*)$ is a dual pair implies in particular that for all $v^*\in\V^*$, the map $v^* \mapsto \inprod{v,v^*}$ is linear.
This is crucial when interpreting $\R\subseteq\V^*$ as the type space, since affine scores must be affine in the type.
For the remainder of this section we assume that we have a dual pair $(\V,\V^*)$.

A natural question is to determine the conditions under which we have $G^{**} \defeq (G^*)^* = G$.  That is, when is the conjugacy operator an involution?  This has been thoroughly studied in convex analysis.  We state the classic theorem due to 
Fenchel and Moreau~\cite{ioffe1979theory,lai1988fenchel}.

\begin{definition}
  \label{def:properties-lsc}
 A function $f:X\to\extreals$ is \emph{lower semi-continuous (\lsc)} if for every $x_0$ in $\dom(f)$ it holds that $\liminf\limits_{x\to x_0} f(x)\ge f(x_0)$.
\end{definition}
\begin{theorem}[Fenchel--Moreau]
  \label{thm:fenchel-moreau}
  Let $X$ be a Hausdorff locally convex space.  For any function $G: X \to \extreals$, it follows that $G = G^{**}$ if and only if one of the following is true: (1) $G$ is a proper, \lsc, and convex function, (2) $G \equiv +\infty$, or (3) $G \equiv -\infty$.
\end{theorem}

The following corollary will prove very helpful in our discussion of type duality below.  The proof follows from applying Theorem~\ref{thm:fenchel-moreau} (note that as $\reals$ is Hausdorff, $\V$ together with the product topology inherited from the dual pair is also Hausdorff and locally convex; see~\cite[\sect~7]{aliprantis2007infinite} for details), and then Lemma~\ref{lem:properties-dual-optimal} twice, once for $G$ and once for $G^*$.

\begin{corollary}
  \label{cor:properties-dual-subgrads}
  If $G$ is convex, proper, and \lsc, then $v^*\in\partial G_v \iff v\in\partial G^*_{v^*}$.
\end{corollary}

We now introduce the concept of a \emph{dual property} $\Gamma^*$, which essentially swaps the type and the report.  That is, an agent has a ``true report'' $r$ and $\Gamma^*(r)$ encodes all the ``correct types'' $t$.  We then go on to show the relationship between the subgradient elicitability of dual properties.

\begin{definition}
  Let $\Gamma:\T\toto \R$ where $\R\subseteq \V^*$.  Then the \emph{dual} of $\Gamma$, written $\Gamma^*:\R\toto \T$, is defined by $\Gamma^* \defeq \Gamma^{-1}$.  In other words, $\Gamma^*$ satisfies $r \in \Gamma(t) \iff t \in \Gamma^*(r)$.
\end{definition}

\begin{theorem}
  \label{thm:dual-type-char}
  For dual pair $(\V,\V^*)$, let $\Gamma:\T\toto \D$ be given with $\T\subseteq\V$ and $\D\subseteq\V^*$.  Let convex proper and \lsc\ $G$ be given.
  Then $G$ elicits $\Gamma$ if and only if $G^*$ elicits $\Gamma^*$.
\end{theorem}
\begin{proof}
  We apply Corollary~\ref{cor:properties-dual-subgrads} to obtain $d\in\partial G_t \iff t \in \partial G^*_d$.  If $G$ subgradient-elicits $\Gamma$, then we have
  \begin{equation*}
    t \in \Gamma^*(d) \iff d \in \Gamma(t) \iff d\in\partial G_t \iff t \in \partial G^*_d,
  \end{equation*}
  so $G^*$ subgradient-elicits $\Gamma^*$.  Clearly the above may be applied in the reverse direction as well, yielding the result.
\end{proof}

Note that when $G$ and $G^*$ elicit $\Gamma$ and $\Gamma^*$, respectively, we have by the above discussion that $\AS{d,t} = \inprod{t,d} - G^*(d)$ elicits $\Gamma$ and $\AS[^*]{t,d} = \inprod{t,d} - G(t)$ elicits $\Gamma^*$.  Moreover, the ``consumer surplus'' functions of $\ASS$ and $\ASS^*$ are $G$ and $G^*$, respectively.  This curious relationship, combined with the notion of report duality, can be visualized as shown in Table~\ref{tab:quadrangle}, the general version of Table~\ref{tab:quadrangle-simple}.
Note that traveling around the table does not necessarily mean arriving at the same choice of $G$, nor does it imply that $G^{**} = G$.  However, when $G^{**} = G$ does hold, the diagram ``commutes'' in a certain sense.

We conclude with a few remarks.

\renewcommand{\myp}{\textcolor{dist}{t}}
\renewcommand{\myq}{\textcolor{util}{d}}
\renewcommand{\mypp}{\textcolor{dist}{t'}}
\renewcommand{\myqq}{\textcolor{util}{d'}}
\begin{table}[t]
  \centering
  \def\dualitycellwidth{160pt}
  \begin{tabular}{c|c|@{}c@{}|@{}c@{}|}
    \multicolumn{1}{r}{} & \multicolumn{1}{r}{} & \multicolumn{2}{c}{Private Type}
    \\[4pt]
    \cline{3-4}
    \multicolumn{1}{r}{} & \multicolumn{1}{r|}{} &
    \textcolor{dist}{Primal ($\myp$)} \bigstrut & \textcolor{util}{Dual ($\myq$)}\\[2pt]
    \cline{2-4}
    \multirow{2}{*}{\parbox[c][100pt]{10pt}{\rotatebox{90}{Report}}} &
    \parbox[c]{10pt}{\rotatebox{90}{\textcolor{dist}{Primal ($\mypp$)}}} &
    \parbox[c][70pt]{\dualitycellwidth}{\centering $\AS{\mypp,\myp}$ $=$ $G(\mypp) + \inprod{\myp-\mypp,dG_{\mypp}}$} &
    \parbox[c][70pt]{\dualitycellwidth}{\centering $\AS[^*]{\mypp,\myq}$ $=$ $\inprod{\mypp,\myq} - G(\mypp)$}\\ 
    \cline{2-4}
    &
    \parbox[c]{10pt}{\rotatebox{90}{\textcolor{util}{Dual ($\myqq$)}}} &
    \parbox[c][70pt]{\dualitycellwidth}{\centering $\AS{\myqq,\myp}$ $=$ $\inprod{\myp,\myqq} - G^*(\myqq)$} &
    \parbox[c][70pt]{\dualitycellwidth}{\centering $\AS[^*]{\myqq,\myq}$ $=$ $G^*(\myq) + \inprod{dG^*_{\myqq},\myq-\myqq}$}\\
    \cline{2-4}\noalign{\vskip 5pt}
    \multicolumn{1}{c}{}& \multicolumn{1}{c}{}&
    \multicolumn{1}{c}{$\sup\: \AS{\cdot,\myp} = G(\myp)$} &
    \multicolumn{1}{c}{$\sup\: \AS[^*]{\cdot,\myq} = G^*(\myq)$}
    \vspace{5pt}
  \end{tabular}
  \caption{The general duality quadrangle.}
  \label{tab:quadrangle}
\end{table}

\subsubsection*{Identities}

Table~\ref{tab:quadrangle-simple} shows that elicitation inherits a lot of structure from convex duality.  Ignoring boundary and regularity concerns for the moment, and looking at general types $t$, we obtain some nice identities:
\begin{align}
  \AS{d,t} + \AS[^*]{t,d} &\geq \inprod{t,d}\\
  \AS{d,t} - \AS[^*]{t,d} &= G(t) - G^*(d).
\end{align}
The first follows from the classic Fenchel-Young inequality~\cite{rockafellar1997convex}, the proof of which for $G$ proper follows directly from the definition of the conjugate (Definition~\ref{def:properties-conjugate}).
\begin{lemma}[Fenchel-Young inequality]
  \label{lem:properties-fenchel-inequality}
   $\forall\; v\in\V,\,v^*\in\V^*$,\; $G(v)+G^*(v^*) \geq \inprod{v,v^*}$.
\end{lemma}

\subsubsection*{Score divergences}

The score divergence $\AS{t,t} - \AS{t',t}$ is a natural notion of ``regret'' which arises frequently in the scoring rules literature (cf. \cite{gneiting2007strictly}).  Our score divergence, as we define below, is reminiscent of a Bregman divergence.
\begin{equation}\label{eq:properties-score-div}
  D_{G,dG}(t,t') \defeq \AS{t,t} - \AS{t',t}  = G(t) - G(t') - \inprod{t-t',dG_{t'}}.
\end{equation}
Note that the first argument to $D$ is the true type, as opposed to our $\ASS$ notation.  Also note the subscripts to $D$, which specify both the convex function $G$ and a selection of subgradients.  A Bregman divergence requires $G$ to be continuously differentiable, but our definition~\eqref{eq:properties-score-div} is a natural extension, and has been studied before (cf.~\cite{iyer2013lovasz}).

Score divergences have many nice properties, like convexity in the first argument, and (directional) differentiability at $t'=t$.
Score divergences also enable reasoning about the magnitude of off-equilibrium payoffs, which can be important in practice, when externalities are often present.  For example, Fiat et al.~\shortcite{fiat2013approaching} introduce the notion of ``strong truthfulness'', where the payoff decays as $\|t-t'\|^2$, to design mechanisms that are robust even when agents care about the utility of other agents.

Turning to our various notions of duality, the following are four divergences corresponding to the duality quadrangle, starting in the (primal,primal) setting and moving counter-clockwise.
\begin{align}
  D_{G,dG}(t,t') &= G(t) - G(t') - \inprod{ t-t',dG_{t'}}\\
  D_{G}(t,d') &= G(t) + G^*(d') - \inprod{t,d'}\\
  D_{G^*,dG^*}(d,d') &= G^*(d) - G^*(d') - \inprod{dG_{d'}^*, d-d'}\\
  D_{G^*}(d,t') &= G^*(d) + G(t') - \inprod{t',d}.
\end{align}
Interestingly, we see that $D_G(t,d) = D_{G^*}(d,t)$ for all $t,d$ (not just dual points).  In other words, the loss of reporting $d$ in the primal but having type $t$ is the same as reporting $t$ in the dual but having ``type'' $d$.

\section{Bregman Voronoi digrams and the role of $\|\cdot\|^2$}
\label{sec:breg-vor}

The squared norm seems fundamental to our derivation in \sect~\ref{sec:props-finite}; let us dig further to see if this is indeed the case.  Observe that the form~\eqref{eq:finite-char} is simply
\begin{align*}
  \label{eq:properties-6}
  \AS{r,t} &= 2\inprod{t_r,t} - \|t_r\|^2 + w_r,
\end{align*}
where $t_r = \alpha p_r + p_0$.  Consider the case where $w_r = 0$ for all $r$, which corresponds to $\Gamma$ being a \emph{Voronoi diagram}.  In this case, could think of $\ASS$ as being a special case of the ``Brier score'' $\AS[^B]{t',t} = 2\inprod{t,'t} - \|t'\|^2$, so that $\AS{r,t} = \AS[^B]{t_r,t}$.  In other words, we can think of our finite-report case as just restricting the allowed reports in a general direct-revelation \scorename.  Note that the score divergence for $\ASS^B$ is just $D_G(t',t) = \|t'-t\|^2$, where $G(t) = \|t\|^2$ is just the square norm.
This raises the following interesting question: what do we get when we replace $G = \|\cdot\|^2$ with another convex function on $\T$, and restrict the reports from $\T$ to just a few points $\{t_r\}_\R$?  That is, take $\AS[^G]{t',t} = G(t') - dG_{t'}(t-t')$ and set $\AS{r,t} = \AS[^G]{t_r,t}$.  Surely, for any such $G$, whatever $\Gamma$ is elicited by such a modified \scorename would have to be a diagram by Theorem~\ref{thm:power-diag}.  But then why does the squared norm seem so fundamental?

As it happens, we are touching on precisely the notion of a \emph{Bregman Voronoi diagram}, introduced by Boissonnat et al.~\citeyear[\sect~4]{boissonnat2007bregman}.  There, instead of defining $\mathrm{cell}_i = \{ t : i\in\argmin_j \|t_j - t\| \}$, the squared norm is replaced by any Bregman divergence $D_G$, so that $\mathrm{cell}_i = \{ t : i\in\argmin_j D_G(t,t_j) \}$.\footnote{In Boissonnat et al.~\citeyear{boissonnat2007bregman}, three types of diagrams are introduced; here we refer to the first type.}  Our conclusion that such diagrams coincide with power diagrams corresponds to their Theorem 8.

Framed in terms of our report duality from \sect~\ref{chap:propertiesreport-duality}, we can see this yet another way.  We can rewrite the Bregman Voronoi cell as
\begin{equation}
  \label{eq:properties-7a}
  \mathrm{cell}_i = \left\{ t : i\in\argmax_j \; G(t_j) + dG_{t_j}(t-t_j) \right\}.
\end{equation}
By Lemma~\ref{lem:properties-dual-optimal}, this can in turn be written
\begin{equation}
  \label{eq:properties-7b}
  \mathrm{cell}_i = \left\{ t : i\in\argmax_j \; \inprod{\tilde t_j,t} - G^*(\tilde t_j) \right\},
\end{equation}
where $\tilde t_j = dG_{t_j}$.  Hence, for any convex function $G$, the sites $\{p_j\}$ and weights $w$ of a power diagram corresponding to the $D_G$ Bregman Voronoi diagram with sites $\{t_j\}$ are given by $p_j= \frac 1 2 dG_{t_j}$ and $w_j = \frac 1 4 \|dG_{t_j}\|^2 - G^*(dG_{t_j})$.

\section*{Acknowledgments}
We would like to thank
Aaron Archer,
Gabriel Carroll,
Yiling Chen,
Philip Dawid,
Felix Fischer,
Hu Fu,
Philippe Jehiel,
Nicolas Lambert,
Peter Key,
David Parkes,
Colin Rowat,
Greg Stoddard,
Matus Telgarsky,
and
Rakesh Vohra
for helpful discussions,
and anonymous reviewers for helpful feedback.
This work was supported in part by National Science Foundation Grant CCF-1657598.

\let\section\oldsection
\bibliographystyle{acmsmall}
{
\bibliography{diss,extra}
}

\end{document}